\documentclass[mnsc,nonblindrev]{informs4}

\OneAndAHalfSpacedXI 

\usepackage{algorithm}
\usepackage{algpseudocode}

\usepackage{subcaption}
\usepackage{mathrsfs}
\usepackage{graphicx}
\usepackage{url}
\definecolor{DarkBlue}{rgb}{0,0.08,0.45}
\usepackage[backref = false, bookmarks, colorlinks = true, plainpages = false, citecolor = DarkBlue , urlcolor = DarkBlue, filecolor = DarkBlue, linkcolor = DarkBlue, hypertexnames=false]{hyperref}
\usepackage{multirow}
\usepackage{newunicodechar}
\usepackage[utf8]{inputenc}
\usepackage{natbib}
 \bibpunct[, ]{(}{)}{,}{a}{}{,}%
 \def\bibfont{\small}%
\TheoremsNumberedThrough     
\ECRepeatTheorems

\EquationsNumberedThrough    

\begin{document}
\newcommand{\abs}[1]{\left|  #1 \right| }
\newcommand{\brak}[1]{\left(#1\right)}    
\newcommand{\crl}[1]{\left\{#1\right\}}   
\newcommand{\edg}[1]{\left[#1\right]}     
\newcommand{\norm}[1]{\|#1\|}
\newcommand{\floor}[1]{\lfloor #1 \rfloor}

\newcommand{\cA}{{\mathcal A}}
\newcommand{\cB}{{\mathcal B}}
\newcommand{\cD}{{\mathcal D}}
\newcommand{\cF}{{\mathcal F}}
\newcommand{\cG}{{\mathcal G}}
\newcommand{\cH}{{\mathcal H}}
\newcommand{\cK}{{\mathcal K}}
\newcommand{\cL}{{\mathcal L}}
\newcommand{\cM}{{\mathcal M}}
\newcommand{\cR}{{\mathcal R}}
\newcommand{\cS}{{\mathcal S}}
\newcommand{\cT}{{\mathcal T}}
\newcommand{\cX}{{\mathcal X}}
\newcommand{\cP}{{\mathcal P}}

\newcommand{\mA}{{\mathbb A}}
\newcommand{\mV}{{\mathbb V}}
\newcommand{\mC}{{\mathbb C}}
\newcommand{\mR}{{\mathbb R}}
\newcommand{\mE}{{\mathbb E}}

\newcommand{\bb}{{\mathbf b}}
\newcommand{\bp}{{\mathbf p}}
\newcommand{\bc}{{\mathbf c}}
\newcommand{\bg}{{\mathbf g}}
\newcommand{\bl}{{\mathbf l}}
\newcommand{\br}{{\mathbf r}}
\newcommand{\bw}{{\mathbf w}}
\newcommand{\by}{{\mathbf y}}
\newcommand{\bx}{{\mathbf x}}

\newcommand{\bA}{{\mathbf A}}
\newcommand{\bB}{{\mathbf B}}
\newcommand{\bC}{{\mathbf C}}
\newcommand{\bD}{{\mathbf D}}
\newcommand{\bG}{{\mathbf G}}
\newcommand{\bL}{{\mathbf L}}
\newcommand{\bS}{{\mathbf S}}
\newcommand{\bQ}{{\mathbf Q}}
\newcommand{\bU}{{\mathbf U}}
\newcommand{\bV}{{\mathbf V}}
\newcommand{\bX}{{\mathbf X}}
\newcommand{\bZ}{{\mathbf Z}}
\newcommand{\bF}{{\mathbf F}}

\newcommand{\bmu}{{\boldsymbol \mu}}
\newcommand{\balpha}{{\boldsymbol \alpha}}
\newcommand{\bbeta}{{\boldsymbol \beta}}
\newcommand{\bepsilon}{{\boldsymbol{\varepsilon}}}
\newcommand{\btheta}{\boldsymbol{\theta}}
\newcommand{\bdelta}{\boldsymbol{\delta}}
\newcommand{\bsigma}{\boldsymbol{\sigma}}
\newcommand{\bnu}{\boldsymbol{\nu}}
\newcommand{\bSigma}{\boldsymbol{\Sigma}}
\newcommand{\bgamma}{\boldsymbol{\gamma}}
\newcommand{\bs}{\boldsymbol{s}}
\newcommand{\bz}{\boldsymbol{z}}

\newcommand{\E}{\mathbb{E}}
\newcommand{\F}{\mathbb{F}}
\newcommand{\p}{\mathbb{P}}
\newcommand{\Q}{\mathbb{Q}}

\newcommand{\R}{\mathbb{R}}
\newcommand{\q}{\mathbb{Q}}

\newcommand{\tr}{{\rm tr}}

\newcommand{\id}{{\mathbbm 1}}

\newcommand{\expect}{\mathbb{E}}

\newcommand{\III}{{\mathcal{I}}}
\newcommand{\HHH}{{\mathcal{H}}}
\newcommand{\LLL}{{\mathcal{L}}}
\newcommand{\CCC}{{\mathcal{C}}}
\newcommand{\XXX}{{\mathcal{X}}}
\newcommand{\YYY}{{\mathcal{Y}}}
\newcommand{\bh}{{\mathbf{h}}}
\newcommand{\bq}{{\mathbf{q}}}
\newcommand{\bk}{{\mathbf{k}}}
\newcommand{\bv}{{\mathbf{v}}}
\newcommand{\bu}{{\mathbf{u}}}

\newcommand{\softmax}{{\text{softmax}}}


\RUNAUTHOR{}

\RUNTITLE{LLM Biases}

\TITLE{LLM Biases}

\ARTICLEAUTHORS{%
\AUTHOR{Jinhui Han}
\AFF{Guanghua School of Management, Peking University, Beijing 100871, China, \EMAIL{jinhui.han@gsm.pku.edu.cn}}
\AUTHOR{Ming Hu}
\AFF{Rotman School of Management, University of Toronto, Toronto, Ontario, Canada M5S 3E6, \EMAIL{ming.hu@rotman.utoronto.ca}}
\AUTHOR{Xilin Zhang}
\AFF{Rotman School of Management, University of Toronto, Toronto, Ontario, Canada M5S 3E6, \EMAIL{xilin.zhang@utoronto.ca}}
} 

\ABSTRACT{Transformer-based agentic AI is rapidly being deployed on major platforms to help users shop, watch, and navigate content with less effort. While these systems can deliver impressive performance, a key concern is whether they may be less reliable than they appear. We ask a simple but fundamental question: whether the mechanisms that make transformer-based agents effective can also induce systematic biases or distortions? We study this question through a theoretical analysis of transformer-based generative recommenders, in which the next user interaction is generated sequentially from a user’s history. Focusing on how the model allocates attention across historical evidence, we identify four bias channels: (i) \emph{Positional bias:} stronger positional encoding shifts influence toward recent history, improving responsiveness but potentially reducing stability and long-term diversity; (ii) \emph{Popularity amplification:} small frequency differences in data can be magnified into disproportionate exposure, contributing to Matthew effects and echo chambers; (iii) \emph{Latent driver bias:} when important drivers of user choices are not directly observed, the model can place overly concentrated weight on a small subset of past events, creating overconfident attributions. (iv) \emph{Synthetic data bias:} when users increasingly follow AI suggestions and platforms retrain on model-shaped synthetic logs, outputs can concentrate over time, and long-tail alternatives can disappear first. Our analysis highlights mechanism-level reliability risks that may not be visible in offline performance metrics. The four bias channels indicate that large-scale deployment may systematically distort exposure and choice. For managers, the immediate implication is to treat these as operational risk factors and to monitor concentration and drift over time, rather than assuming that performance gains alone guarantee reliability.

}%



\maketitle

%


\section{Introduction}

With the rapid advances and broad deployment of large language models (LLMs), firms are increasingly adopting transformer-based models and LLM-style agents to tackle large-scale, complex decision problems. In Operations Management (OM) and Operations Research (OR), this progress has accelerated a shift from using AI primarily for prediction (e.g., demand forecasting or scoring) to using AI as a decision policy that maps rich histories into actions. A growing body of OM research and emerging business practice now explores transformer and agentic architectures for core decisions in assortment planning \citep{wang2023transformer, peng2024transformer}, inventory management \citep{duan2025ask,zhang2026llm, baek2026ai}, and supply chain planning and management \citep{simchi2025large, yin2025rethinking}. This development is particularly timely because customers and markets are increasingly delegating decisions to AI in real-world commerce journeys, reshaping how demand is formed and observed \citep{allouah2025your, baek2026evaluating}.

Against this broader backdrop, LLM-style and transformer-based AI agents are reshaping consumer platforms as shopping assistants, streaming copilots, and content navigators. Each time a platform decides what to show next or where to route a user, through a ranked list, a feed slot, an item page, or a bundle, it is effectively allocating a scarce exposure budget across items, sellers, and categories. In this sense, agentic recommendation is no longer just a back-end ranking tweak, but it has become a primary mechanism through which platforms convert attention into demand. Reflecting this shift, major platforms are investing heavily in transformer-based recommendation capabilities. YouTube, for example, states that recommendations operate primarily on the homepage and the ``Up Next'' panel \citep{Goodrow2021YouTubeRecSys} and develops industrial-scale generative recommendation methods \citep{he2025plum}. In music streaming, Spotify positions its AI DJ as a personalized guide that selects what to play next \citep{Spotify2025DJ}, and Spotify Research discusses semantic identifiers to support generative search and recommendation \citep{SpotifyResearch2025SemanticIDs}. Beyond streaming, Pinterest reports a transformer-based real-time user action model deployed on its largest engagement surface, Homefeed \citep{xia2023transact}. Amazon describes Rufus as a shopping assistant that helps customers discover products and supports intent-based shopping flows \citep{Amazon2025RufusMore}. OpenAI has rolled out shopping assistants in ChatGPT, including Pulse and ``Buy it in ChatGPT'', which are explicitly designed to support product research and recommend relevant options \citep{OpenAI2025ShoppingResearch}. Moreover, Meta describes GEM as an LLM-inspired ``foundation model'' designed for ads recommendation \citep{Meta2025GEM}.

What makes LLM-style and transformer-based AI agents powerful, relative to many traditional generative models, is the \emph{attention mechanism}, which enables the model to sift through a long history (past interactions, context, and instructions) and focus on the most relevant elements for the next prediction. Despite their impressive performance in reducing search and decision frictions, a central question for large-scale adoption remains: \emph{do transformer-based/LLM-style agents introduce systematic biases in recommendations?} We take a step toward answering this question by focusing on transformer-based generative recommenders (rather than fully dialogue-driven assistants), which preserves analytical tractability while maintaining practical relevance. In this framework, given a user’s interaction history, the model generates a short token code that maps to the next recommended item. The code can be learned from rich item information, thereby reflecting text, images, and structured attributes rather than a simple item ID. Importantly, generative recommenders can be deployed as the decision engine behind shopping assistants, streaming copilots, and content agents, which continuously produce personalized actions rather than static rankings.

Building on this setup, we establish four practically meaningful biases endogenously induced by the transformer mechanism: (i) \emph{Positional bias}. Transformers almost always include information about where each past interaction appears in the sequence, either its absolute position or, more often, relative position, indicating how far it is from the current step. We show that relative positional encodings can systematically put more weight on certain parts of a user’s history, so that recommendations tilt toward older versus newer behavior depending on the effective strength of the positional signal. (ii) \emph{Popularity bias}. Transformers can yield super-linear exposure amplification, with two distinct manifestations: (a) a platform-level \emph{Matthew effect} in which already popular items receive disproportionate exposure, increasing sales/play/traffic concentration; and (b) a user-level \emph{echo-chamber} effect that reinforces within-user frequency patterns. (iii) \emph{Latent driver bias}. Generative recommenders learn from clicks, views, and purchases. Yet many choices are decided by latent drivers not recorded in the data (e.g., situational context, user mood). We capture this mismatch as stochastic perturbations in representations and show that, as uncertainty increases, attention can become brittle and unreliable. (iv) \emph{Synthetic data bias}. As users increasingly delegate decisions to AI, interaction logs become endogenously shaped by the agent rather than generated by independent users. When platforms repeatedly retrain on these synthetic logs, recommendations can drift toward a narrower portion of the product catalog.

We emphasize that the four biases we have identified are not arbitrarily chosen. Each aligns with recognized concerns in recommender systems and with emerging issues in rapidly expanding agentic AI decision environments. Positional bias and popularity bias, for example, are widely discussed as model-intrinsic sources of bias in recommendation \citep{chen2023bias, dai2024bias}. On the other hand, misspecification of latent uncertainty and synthetic-data bias under delegated decision-making have become central concerns as agentic systems begin to shape the data on which they are trained and evaluated. Our paper provides complementary theoretical results to the growing body of empirical evidence that these biases can arise from the transformer’s attention mechanism, even when the model is trained ``correctly'' on historical logs. Our results offer practical insights for platform managers on when and how to deploy generative recommenders and agentic assistants, and what risks to anticipate in real-world implementations. 

Finally, note that while our analysis focuses on transformer-based generative recommenders, the insights extend naturally to many agentic AI deployments. A wide range of agents, whether they recommend items, draft actions, or execute multi-step tasks, operate via sequential generation. They repeatedly condition on an evolving context (e.g., interaction history, retrieved information, and prior outputs), and then produce the next token, action, or choice. In such systems, the same attention mechanism determines how past signals are weighted at each step. Consequently, the biases we characterize are not artifacts of a narrow recommender setting, but they are plausible risks in broader agentic settings whenever decisions are mediated by sequential generation. In this sense, our results offer a tractable foundation for reasoning about bias and governance in more general dialogue and action-driven agents.

The remainder of the paper is organized as follows. Section \ref{sec:lit} reviews related work on agentic AI, generative recommenders, and transformer bias. Section \ref{sec:model} introduces the generative recommender architecture that underpins our analysis. Section \ref{sec:bias} presents our main results on the biases induced by transformer mechanisms. Section \ref{sec:conclusion} concludes our findings and discusses future directions.

\section{Literature Review}\label{sec:lit}

Our literature review is organized around three related directions. First, we review the growing OM and OR literature that treats transformers and LLMs as decision-making agents for operational decision problems, including supply chain, inventory, and pricing. Second, we review transformer-based generative recommenders, with emphasis on semantic ID and tokenization-based models that are well-suited to analytical study. Third, we discuss emerging evidence of transformer-induced biases, drawing on recommender systems and broader transformer settings.

\subsection{Agentic AI in Operations Management}

A rapidly growing OM literature is repositioning transformers and LLMs from analytics tools to end-to-end decision-making agents, thereby reshaping not only prediction and automation but also the framing, execution, monitoring, and governance of operational decisions \citep{cohen2026om, dai2025assured}. This broader shift is visible across a range of OM problems. 

In assortment planning and
choice modeling, transformer-based architectures are used to capture richer behavioral structure
than classical discrete-choice baselines \citep{peng2024transformer, wang2023transformer}. For supply chain management, LLMs are developed to facilitate the understanding of data-driven planning tools \citep{simchi2025large, yin2025rethinking}. In pricing and revenue management, several papers \citep{cao2026llm, keppo2025ai} show that delegating pricing to LLM agents can yield supracompetitive outcomes and lead to autonomous collusion under repeated interaction. \citet{jiang2025rideagent} develop an LLM-enhanced feature-driven optimization framework for taxi fleet pre-allocation and pricing. A parallel stream focuses on turning managerial intent into formal optimization models and solver-ready code. ORLM \citep{huang2025orlm} and OptiMUS \citep{ahmaditeshnizi2024optimus} aim to map natural-language problem descriptions into mathematical formulations, executable code, and iterative debugging loops.

Within inventory problems, several papers use the newsvendor setting as a clean benchmark to understand how LLM-based agents reason under uncertainty and how their decisions depart from optimal policies. \citet{wang2024understanding} explore pretrained transformers as decision rules for dynamic pricing and newsvendor problems. Some papers treat LLMs less as stand-alone solvers and more as collaborators that elicit inputs, clarify ambiguous goals, invoke rigorous algorithms \citep{duan2025ask}, and complement human and OR methods \citep{baek2026ai}. \citet{liu2025large} report persistent deviations and bias amplification patterns across leading LLM models even when the optimal formula is provided. \citet{wang2025human} study richer frictions in human-AI interactions, such as partial adherence, behavioral bias, and learning in human-AI-human newsvendor systems. 

Another closely related stream uses LLMs not to act directly, but to generate synthetic preferences or survey responses that then feed downstream decisions. This includes work on LLM-persona-generated distributions for assortment, pricing, and newsvendor decisions \citep{baek2026evaluating}, as well as market-research papers that study data augmentation and the optimal amount of LLM-generated synthetic data to blend with real observations \citep{wang2024large, yin2025rethinking}. \citet{hao2025voice} show that the way how AI agents are presented and interact with users can materially shape user responses. \citet{bhat2024evaluating} show that personalized value alignment in human-robot interaction can affect trust and team performance outcomes. This line of work suggests that AI systems increasingly shape not only upstream operational data on which future actions may be trained, but also the behavioral responses through which such data are generated.

Taken together, these papers imply a new trend in OM: from optimization and prediction as back-end tools to agentic systems that interpret context, interact with users, generate synthetic inputs, propose actions, and, in some cases, act autonomously in core operational decisions. Meanwhile, the same move that creates speed and scale also elevates the importance of understanding how transformer-based mechanisms translate history, context, and feedback into actions, because these systems increasingly sit on the critical path of operational performance and market outcomes.

\subsection{Transformer-Based Generative Recommenders}

Transformer models are reshaping recommendation systems, in large part because they scale to massive catalogs and can incorporate heterogeneous item information beyond sparse interaction logs. For comprehensive overviews of this rapidly growing literature, we refer readers to the surveys by \cite{wu2024survey} and \cite{lin2025can}. Surveys typically distinguish between (i) language-based reformulations that cast recommendation as text generation or text understanding, and (ii) ID-based tokenization/indexing that adapts language models to discrete item universes. The language-based line is extensive and has produced a series of influential frameworks that treat recommendation tasks in a unified natural-language processing framework. Representative examples include P5 \citep{geng2022recommendation}, LLM-Rec \citep{lyu2024llm}, and LLaRa \citep{liao2024llara}. 

Our main focus is on the ID-based paradigm, particularly because it yields analytically tractable formulations while retaining the sequential, attention-based structure. Early work in this direction applies transformer attention to sequential recommendation using standard item IDs. \citet{kang2018self} propose SASRec, showing that self-attention can effectively summarize user histories and outperform RNN-style sequence models. \citet{sun2019bert4rec} further advance this line by moving from left-to-right prediction to bidirectional sequence modeling. Subsequent models refine this backbone along different dimensions, such as explicit time-interval encoding in TiSASRec \citep{li2020time} and contrastive regularization to mitigate representation degeneration in DuoRec \citep{qiu2022contrastive}.

More recent work advances toward semantic ID-based generative recommendation, in which each catalog item is mapped to a short sequence of discrete tokens (rather than a single token, as in standard item IDs). A key advantage of this representation is that semantic IDs can be learned from item embeddings that already fuse heterogeneous information, including textual descriptions, structured attributes, images, and collaborative signals. Exemplified by TIGER \citep{rajput2023recommender}, which generates the next item by predicting its semantic ID token by token from the user context, this design replaces or complements large item embedding tables with a compact discrete vocabulary and enables recommendation via sequence generation. Related work further argues that semantic IDs can improve generalization by leveraging statistical dependencies across similar items, making token design and indexing structures central to performance \citep{singh2024better}. A fast-growing follow-up literature therefore focuses on the tokenization step itself, proposing more informative and robust semantic ID constructions through improved quantization schemes, contrastive objectives, and joint training of the tokenizer and the generative model \citep{zheng2024adapting, zhu2024collaborative}. In addition, hybrid architectures have emerged that integrate behavioral interaction signals with semantic representations within a unified generative pipeline, aiming to preserve complementary information sources rather than collapsing them prematurely \citep{wang2024eager, li2023diffurec}.

In sum, the semantic ID paradigm provides a particularly appealing foundation for analytically tractable generative recommendation. It retains the sequential generation structure of transformers while introducing a compact discrete representation that can be learned from rich item embeddings, including embeddings derived from textual descriptions and other metadata. Our theoretical analysis is therefore centered on this powerful and tractable framework.

\subsection{Transformer Biases}

A growing body of literature cautions that the deployment of transformers can introduce biases and risks despite their overall good performance \citep{allouah2025your, castro2025does}. A recent survey of bias in recommender systems in the LLM era highlights positional bias, popularity bias in ID-based recommenders, and hallucination issues in attention-based model development \citep{dai2024bias}. Building on this, we further study two concerns that are increasingly salient in agentic LLM deployments, latent driver bias and synthetic bias. Although our setting is on recommendation, we draw on broader evidence from transformers, as the same attention mechanism is central to these phenomena.

Positional bias is a widely documented phenomenon in which transformer outputs systematically depend on the position of relevant evidence within the input context. Empirically, LLMs often under-utilize information placed in the middle of long contexts \citep{liu2024lost}, and several studies identify ``attention sinks'', where attention mass concentrates disproportionately on early tokens \citep{xiao2023efficient}. Recent work further develops mechanistic accounts and mitigation approaches for such position dependence \citep{wang2025eliminating, wuemergence}. Popularity bias is another central concern. Transformer training and inference can systematically favor frequent patterns, increasing concentration \citep{wu2025generative}. In recommender settings, these forces map naturally to popularity amplification and long-tail suppression \citep{he2023large}.

As agentic AI becomes more prevalent, a separate question is whether these systems are reliable when observed data are noisy with latency or incompleteness \citep{liang2025artificial}. In transformer settings, recent work diagnoses attention-entropy collapse as a training instability that can lead models to over-concentrate on a narrow subset of signals under uncertainty \citep{zhai2023stabilizing}. Finally, a rapidly expanding literature emphasizes that the most consequential distortions may arise in closed-loop environments, where decisions are delegated to agents and models are retrained on agent-shaped data. Such feedback can shift market outcomes, including collusion-like patterns in pricing settings \citep{keppo2025ai,cao2026llm}. Industry evidence from large-scale recommender deployments further underscores the practical challenges of managing feedback loops in production \citep{tong2023navigating}, and related work has shown degradation when models are repeatedly trained on synthetic or model-generated data \citep{seddik2024bad}.

Taken together, these streams of literature establish that the bias concerns we discuss are recurring phenomena across transformers, recommender systems, and LLM-enabled information retrieval pipelines. Our work contributes to this line by providing a mechanism-based, analytically tractable account of how these biases can endogenously arise in transformer-based generative recommendation, and by clarifying the operational regimes in which they are most likely to become pronounced.

\section{Generative Recommendation Model}\label{sec:model}

In this section, we introduce the generative recommender as an autoregressive transformer (i.e., predicting the next item step by step, using only what has been seen/generated so far) and describe its attention structure, highlighting how attention assigns weights to past interactions when producing each next-item prediction.

\subsection{Model}
We consider a platform with an item catalog $\III$, where each item is indexed by $i \in \III$. A user's interaction history up to time $t$ is an ordered sequence $H_t = (i_1, \ldots, i_t)$, where $i_\tau \in \III$ denotes the item clicked, viewed, or purchased at step $\tau$. Given $H_t$, the recommender system selects (or generates) a next item to display to the user.

\textbf{Tokenization and IDs.} For a generative recommender to output items step by step, we represent each item with a code, which is a short token ID sequence. A token is a discrete ID from a fixed list/vocabulary, analogous to a barcode but with richer information. Formally, each item $i$ is assigned a token sequence $C(i) = (c_1, \ldots, c_L) \in \CCC^L$, where $\CCC$ is a finite token vocabulary and $L \geq 1$ is the code length (the number of tokens used to represent an item). We write $\phi:\mathcal I \to \CCC^L$ for this item-to-code mapping so that $C(i) = \phi(i)$. The recommender generates the tokens sequentially, and the generated code is finally mapped back to a catalog item.

The above notation unifies two common ID schemes: (i) \emph{item-ID} recommenders correspond to $L=1$, where each item is mapped to a single unique token (one ID per item); (ii) \emph{semantic-ID} recommenders correspond to $L>1$, where each item is represented by multiple tokens, forming a multi-token code. The practical motivation for semantic IDs is that the code can be derived from a rich item representation that fuses heterogeneous signals, including text (titles, attributes, reviews), images, structured metadata, and collaborative signals from interaction logs, before being discretized into token sequences. The codes can be constructed to reflect item similarity, with items sharing tokens when their content or behavior is similar.

Throughout, we assume $\phi$ is a bijection between items and token sequences (each item has a unique code and each code identifies a unique item). This assumption is without loss of generality for our analysis. On one hand, if a tokenizer produces collisions (two items receiving the same code), one can append a short disambiguation suffix so that codes become unique. 
On the other hand, to ensure that the model outputs only valid item codes, platforms can enforce catalog-constrained decoding by restricting generation to prefixes that match at least one item code in $\{\phi(i):i \in \III\}$. Under these standard engineering fixes, every generated token sequence corresponds to exactly one item in the catalog.

\textbf{Model output and training objective.} We model recommendation as a step-by-step generation process. Given a user’s past interactions (history) $H_t$, the recommender with parameters $\theta$ defines a probability distribution over token sequences $C=(c_1,\ldots,c_L)$:
\begin{equation}
    p_{\theta}(C|H_t) = \prod^L_{l=1} p_{\theta}(c_{l}|H_t, c_{<l}), \label{eq:p_theta}
\end{equation}
where $c_{<l}:=(c_1, \ldots, c_{l-1})$ are the previously generated token IDs. That is, the model predicts the first token, then the second token given the first, and so on, until it produces a full-length $L$-token code. Hence, the probability of recommending item $i$ is simply the probability of generating its code:
\begin{equation}
    p_{\theta}(i \mid H_t) = p_{\theta}(\phi(i)\mid H_t) = \prod_{l=1}^{L} p_{\theta}\left(c_l(i)\mid H_t, c_{<l}(i)\right). \nonumber
\end{equation}

We next discuss how to derive the recommender $p_{\theta}(i \mid H_t)$ through training. Let $\mathcal{D}$ be the logged training dataset, consisting of pairs $(H_t, i_{t+1})$, where $H_t$ is the user’s past interactions up to time $t$, and $i_{t+1}$ is the next item the user interacted with in the log. The model is trained with cross-entropy loss.
\begin{itemize}
    \item \textit{Item-ID generator}: In the item-ID case, each item is represented by a single token, i.e., $L=1$ and $\phi(i)=c(i)$. Then the model output is $p_{\theta}(i \mid H_t)$ over the catalog, and the next-item log-likelihood objective is
    \begin{equation}
        \mathcal{L}_{\text{IID}}(\theta) = -\sum_{(H_t,i_{t+1})\in\mathcal{D}} \log p_{\theta}(i_{t+1}\mid H_t). \nonumber
    \end{equation}
    This family includes autoregressive/causal (the model can attend to history positions, but not to future positions) sequential recommenders \citep{kang2018self, li2020time, qiu2022contrastive} and masked/cloze (the model can attend to all positions in the sequence but the masked ones) training variants \citep{sun2019bert4rec}.

    \item \textit{Semantic-ID generator}: In the semantic-ID case, each item $i$ is represented by a length-$L$ code sequence $\phi(i) = (c_1(i), \ldots, c_L(i))$. The training objective decomposes across tokens:
    \begin{equation}
        \mathcal{L}_{\text{SID}}(\theta) = -\sum_{(H_t,i_{t+1})\in\mathcal{D}} \sum_{l=1}^{L} \log p_{\theta}\left(c_l(i_{t+1}) \mid H_t, c_{<l}(i_{t+1})\right). \nonumber
    \end{equation}   
    This formulation turns recommendation over a large catalog into prediction over a much smaller token vocabulary at each step, and it naturally supports parameter sharing because similar items may share parts of their codes \citep{rajput2023recommender, wang2024eager}.
\end{itemize}

\subsection{Transformer Architecture}

\subsubsection{High-level structure.}
We study a standard \emph{encoder-decoder} transformer architecture used in many modern generative recommenders. The encoder reads a user’s past interactions (clicks, views, purchases) and turns them into a collection of vectors, one per past event, that summarize the history in context. The decoder then generates the code for the next item, one token at a time. At each generation step, the decoder selects which past events to attend to more before producing the next token.

This decision is made by the transformer’s attention mechanism. A helpful way to understand attention is to view it as a data-driven allocation rule over past information. At a given step, the decoder forms: (i) a query vector $\bq_l$, which represents the ``information need'' at the current position $l$ (what the user seems to want now); (ii) for each past interaction, a key vector $\bk_j$ for history position $j$, which acts like a searchable label describing past interaction $j$ (what happened before); and (iii) a value vector $\bv_j$ for history position $j$, which contains the information we would like to extract from past interaction $j$ if it turns out to be relevant (what past information the recommender may reuse). The model compares the query to each key (via dot products), converts these relevance scores into nonnegative weights that sum to one (via a softmax):
\begin{align}
    Z_{l,j} := \frac{\langle \bq_l, \bk_j\rangle}{\sqrt{d}}, \quad A_{l,j}:= \frac{\exp\left(Z_{l,j}\right)} {\sum_{j'=1}^T \exp\left(Z_{l,j'}\right)}, \nonumber
\end{align}
where $T$ is the number of past tokens. Finally, the transformer takes a weighted average of the values:
\begin{align}
    \text{Attn}_l(\bq_l,\{\bk_j, \bv_j\}^T_{j=1}) := \sum^T_{j=1} A_{l,j} \bv_j. \nonumber
\end{align}
In other words, attention spends a unit weight budget across past interactions and returns the corresponding weighted summary. 

To write this compactly, we stack many queries into a matrix $Q$ (one query per row) and stack the keys and values of all candidate past events into matrices $K$ and $V$ (one key-value pair per row). Then, in one line, scaled dot-product attention is
\begin{equation}
    \text{Attn}(Q,K,V) = \text{softmax}\left(\frac{QK^\top}{\sqrt{d}}\right)V, \label{eq:attn}
\end{equation}
where the softmax operator is applied row-wise so that each row of weights sums to one. From an OM perspective, attention is a soft allocation of a unit weight budget across competing candidates (past interactions). The output is therefore an attention-weighted summary of historical information tailored to the current decision step. This view is central to our bias analysis, as systematic distortions can arise from the formation of these weights via dot products and softmax normalization.

\textbf{Notation.} Throughout the paper, we use $[N]$ to denote the set of $\{1, \ldots, N\}$. Recall that a user’s interaction history is an ordered list of items $H_t=(i_1,\ldots,i_t)$, where each item $i_\tau$ is represented by a length-$L$ discrete code
$\phi(i_\tau)=(c_{\tau,1}, \ldots, c_{\tau,L})$. Concatenating these codes turns the history into a single sequence of $T:=t \times L$ token IDs $(c_{1,1}, \ldots, c_{1,L}, \ldots, c_{t,1}, \ldots, c_{t,L})$. The transformer first maps each token ID to a vector via an embedding lookup, assigning a $d^{\text{enc}}$-dimensional vector to each token in the vocabulary. We write $H^{(0)}=(\bh^{(0)}_1,\ldots,\bh^{(0)}_T)$ for these initial history-token embeddings. The next item is generated as a length-$L$ code $C=(c_1,\ldots,c_L)$, and we denote its initial decoder embeddings by $Y^{(0)} =(\by^{(0)}_1, \ldots, \by^{(0)}_L)$.

Both the encoder and the decoder use the same type of computational block to refine the embeddings iteratively. We index these repeated blocks by the layer number $m\in[M]$, where layer $m$ corresponds to the representation obtained after $m$ rounds of internal processing. On the encoder side, let
\begin{align}
    H^{(m)}=\left(\bh^{(m)}_1, \ldots, \bh^{(m)}_T\right), \quad \bh^{(m)}_j \in \mathbb{R}^{d^{\text{enc}}},
\end{align}
denote the encoder outputs at layer $m$. Here $\bh^{(m)}_j$ is the encoder’s vector summary of the $j$-th history token after incorporating information from other parts of the history and is therefore context-aware.
On the decoder side, at generation step $l$, the decoder maintains a vector state that summarizes both what has been generated so far and what has been retrieved from the history. We write
\begin{align}
    Y^{(m)}=\left(\by^{(m)}_1,\ldots,\by^{(m)}_L\right), \quad \by^{(m)}_l \in \mathbb{R}^{d^{\text{dec}}}, \nonumber
\end{align}
for the decoder’s vector states at layer $m$. These states are then used to produce probabilities over the next discrete token ID $c_l$. 

Finally, note that the encoder and decoder need not have the same number of layers; we use the common index $m$ solely for notational convenience. For readability, we present the single-head attention formulation in the main text. Multi-head attention, residual connections, and layer normalization are standard architectural elements and are omitted because they do not modify the dot-product-softmax structure underlying our analysis. We provide an overall framework in Figure \ref{fig:flow}.

\begin{figure}[ht]
\centering
\includegraphics[width=0.95\textwidth]{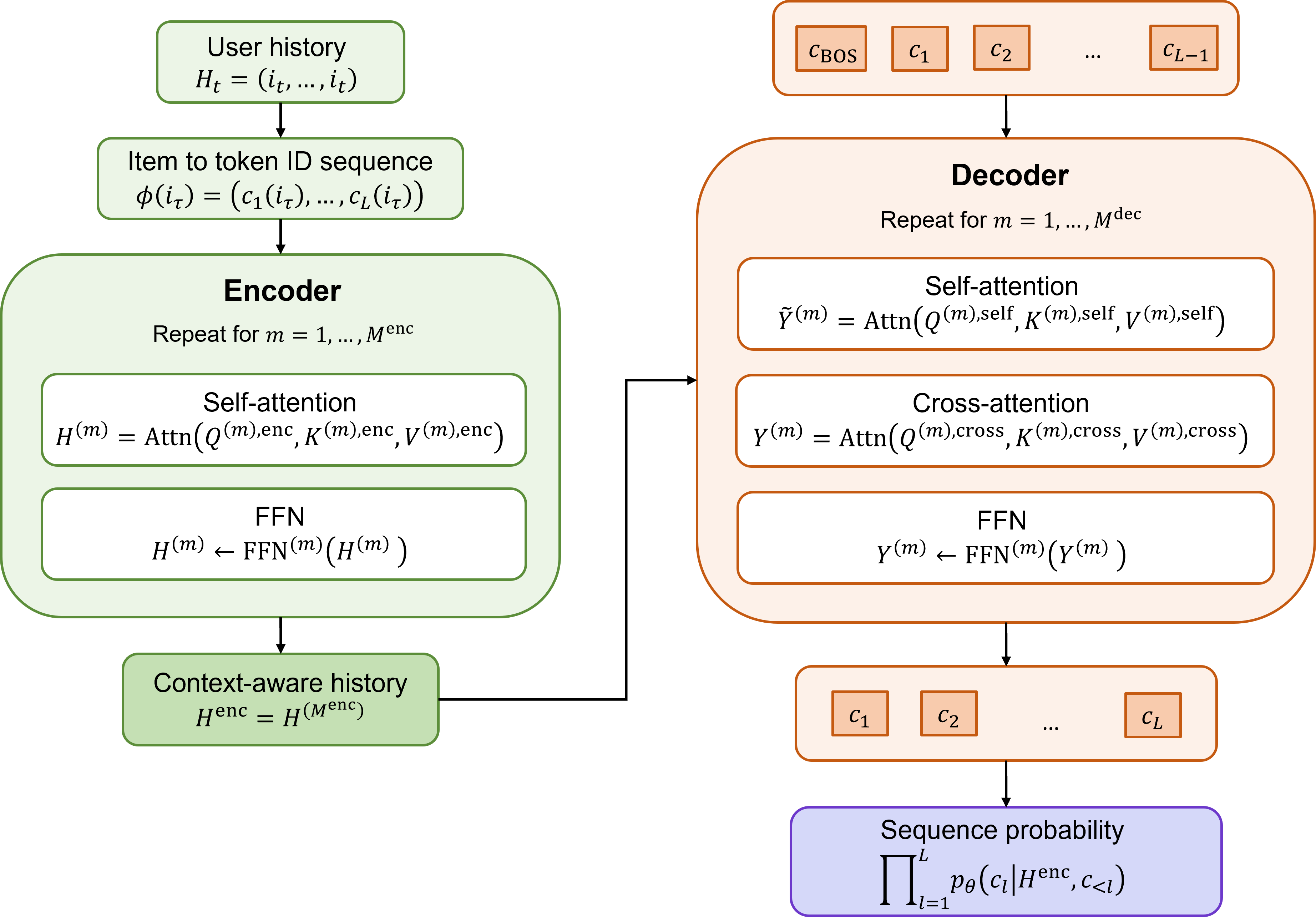}
\caption{Encoder-decoder transformer for generative recommendation}\label{fig:flow}
\end{figure}

\subsubsection{Encoder: self-attention + feed-forward network.}

The encoder converts the raw history embeddings into context-aware history vectors by repeating two steps: (i) let history tokens exchange information (self-attention), and (ii) locally refine each token’s vector (feed-forward network).

\begin{itemize}
    \item[(i)] \textit{Self-attention over the history $H^{(m)}$.} At encoder layer $m$, the input is the previous-layer history matrix $H^{(m-1)}=(\bh^{(m-1)}_1,\ldots,\bh^{(m-1)}_T)\in \mathbb{R}^{T \times d^{\text{enc}}}$. Self-attention creates three derived matrices:
    \begin{equation}
        Q^{(m),\text{enc}} = H^{(m-1)}W^{(m),\text{enc}}_{Q}, \quad K^{(m),\text{enc}} = H^{(m-1)}W^{(m),\text{enc}}_{K}, \quad V^{(m),\text{enc}} = H^{(m-1)}W^{(m),\text{enc}}_{V}, \nonumber
    \end{equation}
    where $W^{(m),\text{enc}}_{Q}, W^{(m),\text{enc}}_{K}, W^{(m),\text{enc}}_{V}$ are trainable matrices that project each history token embedding vector into three roles: \emph{queries} ($Q$) represent what each token is looking for, \emph{keys} ($K$) represent how each token should be matched, and \emph{values} ($V$) represent the information that will be averaged and passed forward. The attention rule then assigns, for each history position $j$, a set of nonnegative weights over all history positions $j'=1, \ldots, T$ (weights sum to one), and replaces each token vector by a weighted average of the value vectors:
    \begin{equation}
        H^{(m)} = \text{Attn}\left(Q^{(m),\text{enc}},K^{(m),\text{enc}},V^{(m),\text{enc}}\right), \nonumber
    \end{equation}
    where $\text{Attn}(\cdot)$ is defined in \eqref{eq:attn}. Here, self-attention performs a soft allocation of a unit-weight budget across past tokens, yielding a context-dependent weighted average. After this step, each element $\bh^{(m)}_j$ summarizes not only the token $j$ itself but also the other parts of the history that the model deems relevant.

    \item[(ii)] \textit{Feed-forward network (FFN).} Attention mixes information \emph{across} positions. The FFN then refines information \emph{within} each position. Concretely, an FFN is a nonlinear function applied to each row independently:
    \begin{equation}
        H^{(m)} \leftarrow \mathrm{FFN}^{(m)}(H^{(m)}). \nonumber
    \end{equation}
    The FFN enriches each token’s representation after pooling information from the rest of the history.
\end{itemize}

After $M^{\text{enc}}$ layers, we denote the final encoder output by
$H^{\text{enc}}:=H^{(M^{\text{enc}})}$. We remark that in some semantic-ID systems, the encoder can be replaced or complemented by other representation modules (e.g., VAE-style tokenizers; \citealt{rajput2023recommender}). We focus on the dot-product softmax attention blocks that appear in transformer components.

\subsubsection{Decoder: self-attention + cross-attention + FFN.}\label{sec:decoder}

The decoder generates the next item’s code token-by-token. At each step, it (i) keeps the partially generated code internally consistent (masked self-attention), 
(ii) retrieves relevant signals from the user’s history (cross-attention), and (iii) applies an FFN to refine the resulting representation.

\begin{itemize}
    \item[(i)] \textit{Self-attention within decoder embeddings $Y$.} At decoder layer $m$, the decoder takes the previous-layer decoder states $Y^{(m-1)} \in \mathbb{R} ^{L\times d^{\text{dec}}}$ and forms
    \begin{equation}
        Q^{(m),\text{self}} = Y^{(m-1)}W^{(m),\text{self}}_{Q}, \quad K^{(m),\text{self}} = Y^{(m-1)}W^{(m),\text{self}}_{K}, \quad V^{(m),\text{self}} = Y^{(m-1)}W^{(m),\text{self}}_{V}. \nonumber
    \end{equation}
    A causal mask (the model can attend to past but not future positions) is applied to enforce the left-to-right rule, i.e., when forming the state at position $l$, the decoder can only use positions $1, \ldots, l$ and cannot look ahead to $l+1, \ldots, L$. Formally,
    \begin{equation}
        \widetilde{Y}^{(m)} = \text{softmax}\left(\text{mask}\left(\frac{Q^{(m),\text{self}}(K^{(m),\text{self}})^\top}{\sqrt{d^{\text{dec}}}}\right)\right)V^{(m),\text{self}}, \nonumber
    \end{equation}
    where $\text{mask}(\cdot)$ sets entries with $k>l$ to $-\infty$ before the softmax (so the corresponding weights become zero). This step allows each code position to borrow information from earlier generated positions, ensuring code coherence.    

    \item[(ii)] \textit{Cross-attention from decoder embeddings $Y$ to history embeddings $H^{\text{enc}}$.} Cross-attention is the mechanism through which the decoder decides which parts of the user’s history matter for generating the next token. Using the masked self-attention output $\widetilde{Y}^{(m)}$ as the decoder’s current state, define
    \begin{equation}
        Q^{(m),\text{cross}} = \widetilde{Y}^{(m)}W^{(m),\text{cross}}_{Q}, \quad K^{(m),\text{cross}} = H^{\text{enc}}W^{(m),\text{cross}}_{K}, \quad V^{(m),\text{cross}} = H^{\text{enc}}W^{(m),\text{cross}}_{V}. \nonumber
    \end{equation}
    Here, the query comes from the decoder (what to decide now), while keys/values come from the encoded history (what happened before). The resulting update is
    \begin{equation}
        Y^{(m)} = \text{Attn}\left(Q^{(m), \text{cross}}, K^{(m),\text{cross}},  V^{(m),\text{cross}}\right). \nonumber
    \end{equation}
    For each decoder position, cross-attention is the key step, allocating weight across all past interactions and producing an attention-weighted summary of the history tailored to the current prediction.
    
    \item[(iii)] \textit{Feed-forward network.} Finally, the decoder applies a position-wise nonlinear refinement:
    \begin{equation}
        Y^{(m)} \leftarrow \mathrm{FFN}^{(m)}(Y^{(m)}). \nonumber
    \end{equation}
\end{itemize}

\textbf{Attention head output.} After the last decoder layer, the model maps the decoder state at step $l$ into a score/logit for each possible next token in the vocabulary $\CCC$:
\begin{align}
    \bu_l := W^{\text{out}} \by^{(M^{\text{dec}})}_l + \bb \in \mathbb R^{|\CCC|}.
\end{align}
These scores are then turned into probabilities by a softmax:
\begin{align}
    p_\theta(c_l = k \mid H_t, c_{<l}) \;=\; \softmax(\bu_l)_k, \quad k\in \CCC.
\end{align}
Intuitively, attention decides how much weight to place on different parts of the past history when forming the output at step $l$, $\by_l$. The output layer then converts $\by_l$ into a probability distribution over the next token ID. Multiplying these step-by-step probabilities across $l=1,\ldots,L$ gives $p_\theta(C\mid H_t)$ in \eqref{eq:p_theta}, and hence the probability of each item via its code.

\textbf{Relation to item-ID generators.} The encoder-decoder structure described above is most natural for semantic-ID generation with code length $L>1$. Two simplifications recover commonly used item-ID recommenders. First, when $L=1$ (each item is a single token), the decoder reduces to a one-step prediction, and the masked self-attention within $Y$ becomes trivial. Second, many sequential recommenders adopt an encoder-only architecture (e.g., SASRec), in which the model directly applies (masked) self-attention to the history and outputs next-item scores without an explicit cross-attention block. Our theoretical results focus on the shared dot-product-softmax attention primitive, which governs how past signals are weighted in both semantic-ID and item-ID implementations.

\section{Transformer-Based Recommender Biases}\label{sec:bias}

We analyze four bias channels that are economically meaningful for various online platforms. Section \ref{sec:position} studies \emph{positional bias}, in which positional encoding and attention interact to overweight particular portions of the user history, generating systematic recency or primacy biases. Section \ref{sec:popularity} identifies \emph{popularity bias}, showing how the dot-product-softmax operator inherent in the attention mechanism can generate super-linear exposure concentration and rich-get-richer dynamics at both platform and user levels. Section \ref{sec:cloning} studies \emph{latent driver bias}, where uncertain behavioral drivers can cause attention to become brittle and dominated by a small subset of observed interactions. Finally, Section \ref{sec:delegation} establishes \emph{synthetic data bias}, where repeated delegation and retraining on agent-shaped logs create closed-loop dynamics that can narrow the effective catalog.

In an encoder-decoder generative recommender, the user’s interaction history influences the output probability distribution over tokens $p_{\theta}(c_l \mid H_t, c_{<l})$ primarily through \textbf{cross-attention} from the decoder to the encoder memory $H^{\text{enc}}$. Hence, our analysis mostly focuses on decoder cross-attention for clarity and tractability, but the logic extends to a broad class of transformer recommenders. We emphasize that any bias arising from the dot-product scoring and softmax normalization in this interface can also occur in encoder self-attention (and in encoder-only recommenders such as SASRec). Thus, we restrict our analysis on the cross-attention block, which can be viewed as a convenient minimal setting that cleanly isolates how the dot-product-softmax converts relative scores into an allocation of weight across historical evidence. Other standard transformer components in a full model, including multi-head aggregation, layer normalization, and feed-forward blocks, may attenuate or amplify the effects we identify. Nevertheless, our results establish that these biases can arise already in a minimal attention block. Understanding these source mechanisms is useful because governance interventions must ultimately act on them.

\subsection{Positional Bias: Anchored vs. Myopic Personalization}\label{sec:position}

Positional bias means that a user's past interactions not only influence the next recommendation through their content (token embeddings), but also through their position in the history. In practice, some agents behave as if early parts of the interaction history carry persistent influence, while others quickly chase recent clicks/views/purchases and flush out older interactions. This depends on how the system allocates scarce exposure across a user’s long-run interests versus short-run signals, trading off responsiveness (adapting quickly) and stability (not overreacting to recency). In this section, by isolating how positional information shapes attention weights in a generative recommender, we provide a concrete, mechanism-level explanation for the positional bias in generative recommenders. 

Our results also suggest why the same ``where-it-appears'' effects can arise more broadly in attention-based AI agents. For example, it is well documented that long-context language models can underweight evidence that appears in the middle of a prompt (``lost in the middle'' by \citealp{liu2024lost}), and can assign disproportionate weight to certain tokens that act as persistent ``attention sinks'' \citep{xiao2023efficient}. Empirical studies also report display-position effects, in which AI agents can systematically favor options shown in specific screen locations even when the underlying options are comparable \citep{allouah2025your}. These settings are not the direct object of our formal analysis, but they reflect a similar position-dependent mechanism in transformers.

\textbf{Positional encoding.} Recall that the code length of each item is $L$, and up to time $t$ there are $T = t \times L$ encoder token embeddings $H^{\text{enc}}=(\bh^{\text{enc}}_1, \dots, \bh^{\text{enc}}_T)$. Fix a decoder layer $m \in [M^{\text{dec}}]$ and a decoder position $l \in [L]$. We treat the encoder memory $H^{\text{enc}}$ as fixed vectors and focus on the analysis of one decoder cross-attention head (which critically connects history to the current decoding step $l$). In the following, we suppress layer superscripts $(m), ~\text{cross}$ for readability. 

In Section~\ref{sec:decoder}, we wrote cross-attention in its simplest form, without positional encodings (PEs), since the attention rule itself does not require positional information. Here, we incorporate PEs, as most practical transformer implementations include some form of positional signal. Let $\widetilde{l} := l+T$, and define the query and key vectors
\begin{align}
    \bq_l := \widetilde{\by}^{(m-1)}_l W_Q, \quad \bk_j := \bh^{\mathrm{enc}}_j W_K, \nonumber
\end{align}
where $\widetilde{\by}^{(m-1)}_l = \widetilde{Y}^{(m-1)}_{l,:}$ is the decoder state at step $l$ after masked self-attention. We consider the widely used Relative positional encoding (RPE). Define the attention \emph{logit} (relevance score) from decoder position $l$ to history position $j$ by
\begin{align}
    Z_{l,j}(\alpha) := \underbrace{\frac{\langle \bq_l,\bk_j\rangle}{\sqrt{d^{\text{cross}}}}}_{\text{content score}} + \underbrace{\alpha \cdot b(|\widetilde{l} - j|)}_{\text{positional term}}, \label{eq:RPE_Z}
\end{align}
where $b(\cdot)$ is a prescribed distance function and $\alpha\ge 0$ scales the overall strength of the positional effect. Thus, the positional term depends only on the relative distance $|\widetilde{l}-j|$ between the current decoder position and the history position. A special case is the ALiBi-style RPE, in which $b(x) = -x$.

The attention \emph{weight} on history token $j$ is the softmax of the corresponding logits:
\begin{align}
    A_{l,j}(\alpha) := \frac{\exp\left(Z_{l,j}(\alpha)\right)}{\sum_{j'=1}^T \exp\left(Z_{l,j'}(\alpha)\right)}. \label{eq:RPE_A}
\end{align}
The cross-attention output is then the weighted average of the value vectors:
\begin{align}
    \by_l = \sum_{j=1}^T A_{l,j}(\alpha) \bv_j =
    \sum_{j=1}^T A_{l,j}(\alpha) \bh^{\text{enc}}_j W_V, \label{eq:rPE_cross_output}
\end{align}
where $\bv_j:=\bh^{\text{enc}}_j W_V$ is the value vector. We impose a mild condition satisfied by common RPE schemes in which the position-based boost is larger for closer tokens.

\begin{assumption}[Near positions receive weakly larger positional logits] \label{asp:RPE_monotone}
For all $i \in S$ and $j\in F$, we have $b(|\widetilde{l}-i|) \geq b(|\widetilde{l}-j|)$.
\end{assumption}

\textbf{Measuring positional bias.} To quantify whether the model places more weight on recent (near) history or on older (far) history, fix a set of indices $S \subseteq \{1,\ldots,T\}$ and interpret it as the near window (for example, the most recent few tokens or interactions). Let $F := \{1,\ldots,T\} \setminus S$ denote the complement, corresponding to older history. For either positional encoding, define the total attention mass assigned to the near window by
\begin{align}
    M_S(\alpha) := \sum_{i\in S} A_{l,i}(\alpha).
    \label{eq:rPE_M}
\end{align}
Since the attention weights $A_{l,i}(\alpha)$ in \eqref{eq:RPE_A} sum to one, $M_S(\alpha) \in[0,1]$ is simply the fraction of the decoder's attention budget allocated to the near part of the history when forming the state at step $l$. If $S$ corresponds to the user's most recent behaviors, then a larger value of $M_S(\alpha)$ indicates that the model relies more heavily on recent evidence, whereas a smaller value indicates that attention is spread more broadly across older history. In this sense, $M_S(\alpha)$ serves as a natural positional-bias diagnostic.

\begin{proposition}[Relative PE tilts attention toward near history]
\label{prop:position}
Under Assumption \ref{asp:RPE_monotone}, for any $\alpha>0$,
\begin{align}
    \frac{d}{d\alpha}M_S(\alpha) \geq 0. \nonumber
\end{align}
where $M_S(\alpha)$ is defined in \eqref{eq:rPE_M} and represents the total attention mass assigned to the near token set. That is, increasing the positional strength $\alpha$ makes the model rely more on the near/recent part of a user’s history and less on older behaviors, even when the content-similarity terms are unchanged.
\end{proposition}

The proof of Proposition \ref{prop:position} can be found in Appendix \ref{pf:prop_position}. In addition, we provide position patterns for the rotary PE (RoPE), another common positional encoding, in Appendix \ref{pf:RoPE}.

\textbf{Anchored vs. myopic personalization.} Proposition \ref{prop:position} shows how shifting $\alpha$ trades off responsiveness (reacting to recent signals) versus anchoring (staying consistent with longer-run preferences). In most recommendation deployments, ``near'' positions are the user's most recent clicks/views/purchases. Proposition \ref{prop:position} then formalizes a clean mechanism: $\alpha$ acts as a recency knob.
\begin{itemize}
    \item \emph{Large $\alpha$ (myopic personalization).} The system reacts quickly to short-lived tastes and session-level intent, because recent interactions receive systematically higher attention weight. This can improve short-run conversion (e.g., capitalizing on transient intent), but it can also suppress stable long-run preferences and long-tail diversity.
    
    \item \emph{Small $\alpha$ (anchored personalization).} Attention becomes less position-driven and more content-driven, so older but persistent interests retain influence. This stabilizes recommendations, supports broader exploration, and can reduce overreaction to noise, but may slow adaptation to genuine shifts in preference.
\end{itemize}
From an OM perspective, this is analogous to a classical responsiveness-stability trade-off (e.g., choosing a larger vs. smaller forgetting factor in forecasting/control).

Proposition \ref{prop:position} suggests two practical levers that map directly to platform objectives: (i) tune positional strength $\alpha$ (or an equivalent scaling in implementation) by product surface (e.g., ``home feed'' vs. ``search follow-up''); (ii) choose the shape of $b(\cdot)$ (how fast weights decay with distance) to control how quickly the system forgets. The same responsiveness-stability trade-off matters in any attention-driven AI systems that make sequential decisions from ordered context, which is consistent with empirical findings on long-context LLMs, including ``lost in the
middle'' behavior and attention sinks \citep{liu2024lost, xiao2023efficient}, as well as display-position effects of agentic shopping systems \citep{allouah2025your}. These findings underscore why positional bias can be regarded as a general operational reliability risk in attention-based AI decision systems.

Meanwhile, we argue that some degree of positional bias is essentially unavoidable, since any non-trivial PE (learned or fixed, absolute or relative) makes influence depend on position. To model ordered histories, the architecture must encode order information, and any non-trivial order signal makes otherwise identical interactions exert different influences when they appear at different positions. Only a fully permutation-invariant model (which discards order) would eliminate positional effects, but that would forgo recency and sequential patterns that are central in recommendation and attention-driven systems.

\subsection{Popularity Bias: Matthew Effects and Echo Chambers}\label{sec:popularity}

A recurring empirical concern in recommender systems is the rich-get-richer pattern, in which popular items tend to receive further exposure, while niche items in the long tail gradually disappear from users' view. Related empirical observations on popularity/long-tail biases in transformer-based recommenders are discussed in \cite{he2023large} and \cite{dai2024bias}. While our analysis focuses on recommenders, a similar generative monoculture has also emerged in broader LLM-style settings. For example, \citet{wu2025generative} find that in book reviews, model-generated reviews are biased toward only positive sentiment. This further supports the view that an attention-based weighting rule can systematically favor what is already frequent or salient.

In this section, we provide a mechanism-level explanation of how such popularity amplification can arise inside a single decoder cross-attention head. The logic is intuitive: (i) when training data are imbalanced, early gradient updates tend to give more frequent tokens a representation advantage; and (ii) the dot-product-softmax attention map converts even a modest representation advantage into a disproportionate, and often exponential, exposure advantage.

\textbf{Token exposure.} Fix one decoder cross-attention head, at a given decoder layer $m \in [M^{\text{dec}}]$ and decoding step $l \in [L]$. To streamline notation, we suppress the superscripts of $(m),\text{cross}$, and $\text{enc}$, as well as the subscript $l$. Recall that $T = t\times L$ is the number of history tokens in the encoder memory for the current prediction step. For each history position $j \in [T]$, cross-attention computes an attention logit (relevance score) $Z_j(\bq)$ that measures how relevant token $j$ looks to the current decoder state $\bq$, and an attention weight $A_j(\bq)$ obtained through softmax normalization:
\begin{align}
    Z_j(\bq) = \frac{\langle \bq, \bk_j\rangle}{\sqrt d}, \quad A_j(\bq) = \frac{\exp(Z_j(\bq))}{\sum_{j'=1}^{T}\exp(Z_{j'}(\bq))}, \quad j\in [T]. \label{eq:pop_ZA}
\end{align}
Here $\bq$ is the query vector at this given decoder step (what the model is trying to decide now), and $\bk_j$ is the key vector
for history token $j$ (what the model uses to assess the relevance of token $j$).

Many history positions may correspond to the same token ID, and the frequency with which a token appears in the history reflects its popularity. Let $\mathcal{H}$ denote the set of distinct token IDs that appear in the history window, and let $\bh_j \in \mathcal{H}$ denote the token ID at history position $j$. We define the token exposure (token-level attention mass) as the total attention weight assigned to all occurrences of token $\bh$ \footnote{Throughout the paper, we use boldface for vectors, but when a vector is indexed, we suppress the boldface and write $h$ instead of $\bh$.}:
\begin{align}
    M_h(\bq) := \sum_{j:\,\bh_j=\bh} A_j(\bq), \quad \bh \in \mathcal{H}. \label{eq:pop_M}
\end{align}

\textbf{Measuring popularity.} 
Because the attention weights sum to one, $M_h(\bq) \in [0,1]$ can be interpreted as the fraction of the attention budget allocated to token $\bh \in \mathcal{H}$ at this prediction step. We also write the set of positions where token $\bh$ appears:
\begin{align}
    S_h := \{j \in [T]: \bh_j=\bh\}, \label{eq:pop_Sh}
\end{align}
with $|S_h|$ being the number of occurrences of token $\bh$ in the history. We let 
\begin{align}
    p_h := |S_h|/T \label{eq:pop_p}
\end{align}
denote the frequency of token $\bh$ in the history.

In a recommender, a token that receives more attention mass is more influential because the attention output is a weighted average of value vectors. Thus, $M_h$ provides a direct measure of how much the model is relying on token $\bh$. If attention does not systematically favor any occurrence beyond its presence in the history,
then one would expect $M_h(\bq) \approx |S_h|/T = p_h$, and therefore $M_h(\bq)/M_{h'}(\bq) \approx p_h/p_{h'}$. This motivates comparing the exposure ratio $M_h/M_{h'}$ to the count ratio $p_h/p_{h'}$.

\textbf{Training objective.} For tractability, we use $-\log M_h(\bq)$ as a token-aggregated surrogate training objective for the true training objective $p_\theta(\cdot \mid H_t, c_{<l})$. LLMs and transformer-based recommenders are typically trained by minimizing a cross-entropy objective to predict the next token using stochastic gradient descent (SGD), most commonly with Adam/AdamW, an adaptive variant of SGD. To isolate the effect of history token frequency on token exposure, we study a surrogate objective that directly rewards the head for assigning attention mass to the correct token label. Specifically, at SGD step $n \in \{1, \ldots, N\}$, we observe a query-label pair $(\bq^{(n)},X^{(n)})$ from data (history interactions), where $X^{(n)} \in \mathcal{H}$ is the target token ID and $\bq^{(n)}$ is the corresponding query vector. We then compute the gradient of the loss with respect to the trainable parameters $\bmu_h^{(n)}$ (embeddings associated with token $\bh$), and update the parameters in the loss-reducing direction with step size determined by the learning rate $\eta$. 

To make the dependence on the trainable parameters explicit, we model the key vector of each history occurrence of token $\bh$ by a shared parameter $\bmu_h^{(n)}$. That is, at SGD step $n \in [N]$ (the training horizon $N$ counts how many SGD steps are performed), we let 
\begin{align}
    \bk_j^{(n)} = \bmu_{h_j}^{(n)}, \quad j\in[T]. \label{eq:pop_shared}
\end{align}
Hence, the attention score of position $j$ is $Z_j^{(n)}(\bq)=\langle \bq,\bmu_{h_j}^{(n)}\rangle /\sqrt{d}$, and the token-level exposure for token $\bh$ becomes
\begin{align}
    M_h^{(n)}(\bq; \bmu^{(n)}) = \sum_{j: \bh_j=\bh} \frac{\exp\left(\langle \bq, \bmu_h^{(n)}\rangle / \sqrt{d}\right)}{\sum_{j'=1}^T \exp\left(\langle \bq, \bmu_{h_{j'}}^{(n)}\rangle / \sqrt{d}\right)} = \frac{|S_h|\exp\left(\langle \bq, \bmu_h^{(n)}\rangle / \sqrt{d}\right)}{\sum_{\bh' \in \mathcal{H}} |S_{h'}|\exp\left(\langle \bq, \bmu_{h'}^{(n)}\rangle / \sqrt{d}\right)}. \label{eq:pop_M_n}
\end{align}
For a training example $(\bq^{(n)},X^{(n)})$ at SGD step $n$, we define the surrogate loss regarding token exposure (instead of the actual transformer training objective)
\begin{align}
    \ell_{X^{(n)}}\left(\bq^{(n)}; \bmu^{(n)}\right) := -\log M^{(n)}_{X^{(n)}}\left(\bq^{(n)}; \bmu^{(n)}\right), \label{eq:pop_surrogate_loss}
\end{align}
where the target token $X^{(n)} \in \mathcal{H}$ is drawn from the (possibly imbalanced) training distribution. To focus on the frequency-driven mechanism, we treat $\{\bmu_h\}_{\bh \in \mathcal{H}}$ as the only trainable parameters in this head and update them via SGD:
\begin{align}
    \bmu_h^{(n+1)} = \bmu_h^{(n)} - \eta \nabla_{\bmu_h}\ell_{X^{(n)}}\left(\bq^{(n)}; \bmu^{(n)}\right). \label{eq:pop_sgd_update}
\end{align}

\begin{remark}{\textbf{(Proxy consistency of the surrogate training objective)}} In the full transformer, training minimizes the standard negative log-likelihood (cross-entropy) of the next token under the model distribution $p_\theta(\cdot \mid H_t, c_{<l})$. Our analysis does not replace this training objective. Rather, it isolates a single cross-attention head and examines how frequency imbalance can translate into exposure imbalance via the dot-product-softmax mechanism. The surrogate loss in \eqref{eq:pop_surrogate_loss} uses $M_h(\bq)$ as a tractable proxy for ``how much this head supports token $\bh$'', because the decoder representation entering the output head is a weighted average of history information, and the total weight assigned to occurrences of token $\bh$ is exactly $M_h(\bq)$. Under the shared-key modeling (occurrences of the same token have similar key/value vectors, see \eqref{eq:pop_shared}), increasing $M_h(\bq)$ increases the contribution of token $\bh$ to the decoder state, which tends to raise its logit and hence its predicted probability under the true cross-entropy objective. Therefore, analyzing $-\log M_h(\bq)$ cleanly captures the same objective in the full model.
\end{remark}

\begin{remark}{\textbf{(From token-level amplification to item-level popularity bias)}}
The object of our analysis is the token exposure $M_h(\cdot)$. For semantic-ID recommenders, an item is generated as a token sequence $\phi(i)=(c_1(i), \dots, c_L(i))$. Token-level amplification implies that items composed of more \emph{head tokens} gain a systematic advantage over items that rely on \emph{tail tokens}. Thus, long-tail items can be under-exposed even if the model is not explicitly trained to ``promote best sellers.'' For item-ID recommenders ($L=1$), token-level and item-level popularity coincide.
\end{remark}

To facilitate analysis, we impose the following assumptions, which are natural when different occurrences of the same token yield similar keys (common in both item-ID ($L=1$) and semantic-ID settings).

\begin{assumption}\label{asp:pop}
For each SGD step $n \geq 0$ and training query-label pair $(\bq^{(n)},X^{(n)})$, we assume:

(i) \textbf{Bounded queries:} There exists a finite constant $B_q>0$ such that $\|\bq^{(n)}\|\le B_q$ for all $n \geq 0$. 

(ii) \textbf{I.i.d. query-label pairs:} The training pairs $\{(\bq^{(n)}, X^{(n)})\}_{n \geq 0}$ are i.i.d. from a common joint distribution. The label satisfies
$\Pr(X^{(n)} = \bh) = p_h$ and the conditional query mean $\bw_h := \mathbb{E}[\bq^{(n)} \mid X^{(n)} = \bh]$ exists for all $\bh \in \HHH$.

(iii) \textbf{Symmetric initialization:} $\bmu^{(0)}_h = \bmu^{(0)}$ for all $\bh \in \HHH$, for some common initialization vector $\bmu^{(0)}$.

(iv) \textbf{Early-stage persistence regime:} The frequency-driven advantage created during the early stage of SGD is not fully offset by later updates by the stopping time.
\end{assumption}
Assumption \ref{asp:pop}(i) is a standard regularity condition in transformer-related theoretical analysis. Assumption \ref{asp:pop}(ii) specifies a stationary data-generating law for the query-label pairs used by SGD. In practice, queries are generated from sequential contexts and need not be independent across samples. Here, we adopt an i.i.d. pair model in order to isolate the mechanism-level role of token popularity or frequency $p_h$ in the update dynamics. No requirement between $\bq^{(n)}$ and $X^{(n)}$ is assumed. By the law of total expectation, 
\begin{align}
    \bar{\bw} := \mathbb{E}\left[\bq^{(n)}\right] = \sum_{\bh \in \HHH} \Pr(X^{(n)} = \bh) \mathbb{E}\left[\bq^{(n)} \Big| X^{(n)}=\bh \right] = \sum_{\bh \in \HHH} p_h \bw_h. \label{eq:pop_barw}
\end{align}
Assumption \ref{asp:pop}(iii) is the usual symmetry condition at initialization of SGD training. Before training, token-specific parameters do not yet encode differential exposure advantages. This provides a neutral starting point from which any subsequent popularity-aligned asymmetry must be created by the learning dynamics rather than imposed ex ante. Assumption \ref{asp:pop}(iv) is consistent with common stopping rules based on improvements in the average loss $\ell_{X^{(n)}}$ across tokens. Under such rules, low-frequency tokens (with smaller $p_h$) contribute less to the aggregate loss improvement, so training may stop before they are fully corrected relative to high-frequency tokens. While a general proof of this phenomenon is difficult, prior work suggests that later training need not undo early advantages in SGD learning \citep{chen2023adap, kleinmancritical}.

Under these assumptions, we first show that more frequent tokens receive stronger directional reinforcement during an analytically tractable early-stage window. We then study how this frequency-weighted early advantage can translate into a larger final exposure ratio when it is not fully offset by later training.

\begin{proposition}[Popularity amplification]\label{prop:pop_bias}
Under Assumption \ref{asp:pop}, fix a test query $\bq$. For a training horizon $\widehat{N}$ and any $\bh, \bh' \in \mathcal{H}$ with $|S_h|,|S_{h'}| \geq 1$, define 
\begin{align}
    \text{AR}(\bh,\bh') := \frac{\mathbb{E}\left[M_h(\bq) / M_{h'}(\bq)\right]}{p_h/p_{h'}}, \label{eq:pop_AR_def}
\end{align}
where the expectation is taken over training randomness $\{(\bq^{(n)}, X^{(n)})\}_{n \geq 0}$. Then,
\begin{align}
    \text{AR}(\bh, \bh') \geq \exp\left(\frac{\eta \widehat{N}}{d} \left\langle \bq, p_h(\bw_h - \bar{\bw}) - p_{h'} (\bw_{h'} - \bar{\bw})\right\rangle - \frac{2\xi^{(\widehat{N})} \|\bq\|}{\sqrt{d}}\right), \label{eq:pop_AR_LB}
\end{align}
where $\xi^{(\widehat{N})}$ is a small residual term upper-bounding the effect of later-stage updates. In particular, when query $\bq$ is comparably aligned with $\bw_h - \bar{\bw}$ and $\bw_{h'} - \bar{\bw}$, the final relative exposure is amplified in favor of the more frequent/popular token.
\end{proposition}

The proof of Proposition \ref{prop:pop_bias} can be found in Appendix \ref{sec:pf_pop_bias}, with a discussion of the early-advantage persistence regime (see Remark \ref{rem:early_window}).

\textbf{When the rich get richer.} In generative recommenders, the same amplification characterized in Proposition \ref{prop:pop_bias} occurs at each decoding step. Accordingly, the resulting bias is not merely a one-shot ranking artifact, but it can compound across tokens and sessions, creating a feedback loop in which high exposure increases future observations and, in turn, further strengthens the head. This provides a concrete mechanism for the rich-get-richer phenomenon: early advantages can compound over time and across tokens, increasing the likelihood that already frequent patterns recur. Proposition \ref{prop:pop_bias} formalizes this mechanism by quantifying how dot-product-softmax attention heads can amplify the frequency imbalance in the data. This yields two managerial implications depending on the scope at which the  training distribution is defined:
\begin{itemize}
    \item \emph{Platform-level concentration and the Matthew effect (global training).} When training data are pooled across users, the frequency parameters $p_h$ correspond to global token (or item-ID) popularities. Proposition \ref{prop:pop_bias} then implies that, even in the absence of any explicit ``promote best-sellers'' rule, frequent tokens/items receive disproportionate attention mass and hence higher exposure. Operationally, this predicts a \emph{Matthew effect}: initially popular products, artists, or topics tend to become even more visible over time, leading to greater concentration in sales (shopping platforms), plays (streaming platforms), or traffic (content and news platforms). This reinforces the concern that purely data-driven recommenders can exacerbate demand concentration and reduce long-tail discovery. 

    \item \emph{User-level concentration and echo chambers (personalized training).} In many settings, the training data can be user-specific (or cluster-specific) logs. In that case, the relevant frequency becomes conditional on the user, $p_{u,h}:=\Pr(X= \bh \mid u)$. Applying the same mechanism, with $p_h$ replaced by $p_{u,h}$, implies that the system can amplify within-user frequency patterns. Tokens that dominate a user's own past interactions receive multiplicatively greater attention mass in subsequent generations. Thus, even when global item popularity is balanced across categories, the model may still produce personal \emph{echo chambers}: each user is repeatedly exposed to similar categories, attributes, or narratives because those tokens are overrepresented in that user's training data.
\end{itemize}

We finally note that some degree of popularity bias is hard to avoid in attention-based systems beyond recommendation. As long as (i) the training data are even mildly imbalanced and (ii) the model uses dot-product-softmax weightings, any learned score separation aligned with frequency is translated into multiplicative exposure differences. It is unrealistic to expect such bias to dissipate spontaneously. Therefore, attention-based or LLM-style AI systems may naturally drift toward concentration, homogenization, or monoculture (\citealt{wu2025generative}). Mitigating these tendencies requires deliberate intervention, such as reweighting or reshaping the training data, regularizing the attention sharpness, or imposing explicit exposure constraints, although such interventions may in turn introduce over-smoothing.

\subsection{Latent Driver Bias: Brittle Personalization}\label{sec:cloning}

Recommender systems learn from past actions to predict what a user will want next. In practice, however, the next click, watch, or purchase is often shaped by \emph{latent drivers} such as mood, time pressure, social context, off-platform information, or spontaneous impulses that do not appear in platform logs. When these unobserved factors matter, personalization can become brittle: small, essentially random differences in the observed history may lead to large swings in what the system recommends. This brittleness can manifest as inconsistent recommendations across seemingly similar users or sessions, overreaction to a few idiosyncratic events, and a weaker ability to reliably ``clone'' (reproduce) user behavior from logs alone
\citep{oh2022rank, liang2025artificial}.

In this section, we model the attention logits (relevance scores) inside attention as noisy proxies for ``what mattered'' when the true driver is only partially observed. The dot-product scores may differ slightly purely because of unobserved noise, yet the softmax step converts these additive score differences into multiplicative differences in attention odds. In particular, even when two history tokens have the same mean relevance score, latent-driver noise can cause one token to receive several times more attention than the other, making the recommender behave as if it has found a clear explanation in the logs even when the underlying choice process remains genuinely uncertain.

\textbf{Cross-attention under latent drivers.} 
We focus on a single decoder cross-attention head at a fixed decoder layer $m$ and decoding step $l$. For notational convenience, we omit the superscripts $(m),\text{cross}$ and the subscript $l$. The attention head assigns a weight $p_j$ to each history token $j \in [T]$, where $p_j \geq 0$ and $\sum_{j=1}^T p_j=1$. These weights indicate the extent to which each past event contributes to the decoder’s current decision. A larger $p_j$ indicates that the model relies more on the history token $j$ in the current step. 

Before normalization, the head computes an attention logit (relevance score) $Z_j$ for each history token $Z_j(\bq)=\langle \bq,\bk_j\rangle/\sqrt{d}$ (the same as in \eqref{eq:pop_ZA}), where $\bq$ is the query vector summarizing what the decoder is trying to decide at this step, and $\bk_j$ is the key vector representing history token $j$. Here, we analyze the distributional effects of additive logit perturbations conditional on the query $\bq$, since $\bq$ merely scales the logits. Thus $\bq$ enters only through the resulting variance parameters of the centered logits and is absorbed into $\{\sigma_j^2\}$, which will be introduced shortly. Accordingly, we suppress the explicit dependence on $\bq$ and write the attention weights as
\begin{align}
    p_j := A_j = \frac{\exp(Z_j)}{\sum_{k=1}^T \exp(Z_k)}, \quad j\in [T]. \nonumber
\end{align}
If user actions are partly driven by latent factors (mood, context, off-platform information, etc.), then the model’s relevance scores are inevitably noisy. We capture this by introducing a random perturbation to the logits $Z = \{Z_j\}^T_{j=1}$ \citep{liang2025artificial}. 

\textbf{Noise-only regime.} To isolate how unobserved impulse (modeled as logit noise) propagates through the dot-product-softmax map into attention, we consider a ``noise-only'' regime, as defined in the following Assumption \ref{asp:clone_noise}. Note that this does not claim the model is driven only by noise. Rather, it quantifies the mechanism by which uncertainty in relevance scores can turn into concentrated attention.

\begin{assumption}[Gaussian logit noise]\label{asp:clone_noise}
The logits are independent Gaussian:
\begin{align}
   Z_j \sim \mathcal{N}(0,\sigma_j^2), \quad \text{independent across } j\in[T], \nonumber
\end{align}
and the attention weight vector is given by $\bp=\softmax(Z)$.
\end{assumption}

\textbf{Measuring brittleness.} We measure recommendation brittleness using the pairwise odds ratio $p_j/p_k$, which directly captures how much more the model trusts event $j$ than event $k$. In particular, for any two history tokens $j$ and $k$, the attention ratio satisfies
\begin{align}
    \frac{p_j}{p_k} = e^{Z_j - Z_k},
\end{align}
so the ratio is driven entirely by the random relative perturbation $Z_j - Z_k$, and large values of $p_j/p_k$ arise purely from noise in the latent drivers. The following Proposition \ref{prop:clone_bias} (with proof in Appendix \ref{sec:pf_cloning}) shows that greater latent uncertainty $\{\sigma_j\}^T_{j=1}$ makes the model more likely to place disproportionately large weight on one observed interaction over another.

\begin{proposition}[Uncertainty-driven attention amplification]\label{prop:clone_bias}
Under Assumption \ref{asp:clone_noise}, for history tokens $j, k \in [T]$,
\begin{align}
    \log\left(\frac{p_j}{p_k}\right) \sim \mathcal{N}(0, \sigma_j^2 + \sigma_k^2). \nonumber
\end{align}
Consequently,
\begin{align}
    \mathbb E\left[\frac{p_j}{p_k}\right] = \exp\left(\frac{\sigma_j^2+\sigma_k^2}{2}\right), \quad \text{and} \quad \Pr\left(\frac{p_j}{p_k}\ge c \right) = 1 -\Phi\left(\frac{\log c}{\sqrt{\sigma_j^2+\sigma_k^2}}\right) ~\forall c>1. \nonumber
\end{align}
That is, additive Gaussian perturbations in the logit gap are exponentiated into multiplicative perturbations in relative attention.
\end{proposition}

\textbf{Limits under latent drivers.} Proposition \ref{prop:clone_bias} highlights an intrinsic mismatch between logged choices and the latent drivers that ultimately shape behavior. Softmax is an exponential gate on score differences. Before softmax, uncertainty enters additively through the logit gap $Z_j - Z_k$. After softmax, the same gap becomes the multiplicative odds ratio $e^{Z_j - Z_k}$. For example, $Z_j-Z_k = 1$ yields $p_j/p_k \approx 2.72$ and $Z_j - Z_k = 2$ yields $p_j/p_k \approx 7.39$. Therefore, even when two history tokens have the same mean relevance, unobserved shocks can cause the model to rely several times more on one token than on another. Operationally, the recommender behaves as if a few logged interactions provide a more dominant explanation for the next action than the rest do, even when the underlying preference process is only partially observed. This generates brittle personalization, where the system may fit logged click traces well yet still mis-clone the true decision process, especially in high-uncertainty settings such as new users, cold-start items, or volatile sessions.

More broadly, the same mechanism arises whenever an attention-based AI system must infer intent from an incomplete context. Because softmax exponentiates the latent logit gap, missing situational information does not merely add noise, but it makes large winner-take-most attention allocations in relative attention more likely. For platform design, potential mitigation approaches include collecting richer contextual signals to reduce latent uncertainty, using temperature or other attention-smoothing regularization to dampen excessive pairwise dominance, and incorporating explicit exploration or uncertainty-aware policies when the effective logit noise is high.

\subsection{Synthetic Data Bias: Retraining on AI-Generated Logs}\label{sec:delegation}

The effects of training on synthetic data have been studied extensively, and the underlying risk is not new. What is new and prominent in the LLM era is how easily decision-making can shift from users to agents, as LLM-based assistants are now widely accessible and sufficiently capable that many users follow their suggestions with little additional search. Recent studies have also examined the use of LLM-generated responses in market research, when the gap between human and AI behavior is sufficiently limited or can be appropriately calibrated \citep{wang2024large, yin2026synthetic, baek2026evaluating}. As users increasingly delegate choices/responses to AIs, platform logs are no longer driven primarily by independent human exploration. Instead, a growing share of recorded interactions is agent-shaped data, reflecting what the agent tends to propose. 

Recent LLM evidence shows that when models are repeatedly retrained on such synthetic or agent-generated interactions, output diversity tends to shrink and rare modes are gradually lost over successive rounds of training \citep{shumailov2024ai, alemohammad2023self}. Importantly, this concentration need not simply mirror the original popularity pattern in human logs. It can also reflect the model’s own stylistic preferences or default tendencies. In this section, we provide a simple mechanism-level account of how retraining on agent-shaped data alone can push a system toward progressively concentrated outputs.

\textbf{A closed-loop retraining model.} We analyze retraining within a single context state $j$, which can be seen as a coarse bucket of ``similar decision situations'' regarding history $(H_t, c_{<l})$, e.g., histories and partial outputs that look alike to the model. Within the same bucket, the next choice is among $s$ discrete tokens, and we index these options by $k \in [s]$. Let $\bp^{(0)}_j=(p^{(0)}_{j1},\ldots,p^{(0)}_{js})$ denote the organic conditional distribution in context $j$ (what users would choose without delegating to the AI). At deployment/retraining round $r=0,1,2,\ldots$, the system produces a conditional distribution $\bp^{(r)}_j$ in the same context $j$.

To model delegation, we assume that the training data collected in round $r$ within context $j$ consist of a mixture of two sources:
\begin{itemize}
    \item \emph{Organic (real) data}: $N_j$ fresh user choices drawn i.i.d. from $\bp^{(0)}_j$;
    \item \emph{Delegated (synthetic) data}: for $r \geq 2$, $\widehat{N}_j$ AI-shaped choices drawn i.i.d. from the previously deployed model distribution $\bp^{(r-1)}_j$.
\end{itemize}
We assume both $N_j>1$ and $\widehat{N}_j>1$. Define the delegation share in the logs by
\begin{align}
    \alpha := \frac{\widehat{N}_j}{N_j+\widehat{N}_j} \in (0,1). \label{eq:del_alpha}
\end{align}
A larger $\alpha$ means a larger fraction of the data in context $j$ comes from delegated (AI-shaped) actions. Under standard cross-entropy training, for a fixed context bucket $j$, the best-fitting categorical predictor is the empirical histogram of next outcomes observed in that context. In Appendix \ref{pf:del_empirical}, we formalize this frequency-matching property in our setting. As a result, the retrained conditional distribution $\bp^{(r+1)}_j$ is essentially the empirical token frequency vector computed from the round-$r$ mixture data in context $j$.

\textbf{Concentration metric.} To measure how concentrated the conditional distribution becomes over rounds, we use the squared $\ell_2$ index
\begin{align}
    \|\bp_j^{(r)}\|^2_2 = \sum_{k=1}^s \left(p^{(r)}_{jk}\right)^2. \label{eq:del_S}
\end{align}
This index is intuitive because $\sum_k (p^{(r)}_{jk})^2$ is the probability that two independent draws from $\bp^{(r)}_j$ produce the same outcome. We define $S^{(r)}_j := \mathbb{E}[\|\bp_j^{(r)}\|^2_2]$. The larger $S^{(r)}_j$ is, the more concentrated or less diverse the generation output is. If $\bp^{(r)}_j$ is uniform over $s$ alternatives, then $S^{(r)}_j=1/s$, corresponding to maximal diversity. By contrast, if $\bp^{(r)}_j$ is nearly deterministic, then $S^{(r)}_j \rightarrow 1$, corresponding to maximal concentration. Next, we present our main findings on how $S^{(r)}_j$ evolves as $r$ increases in the following Proposition \ref{prop:del_bias} (proof found in Appendix \ref{sec:pf_del}).

\begin{proposition}[Retraining on delegated data increases concentration]\label{prop:del_bias}
For each context $j$, if $S^{(0)}_j <1$, then $S_j^{(r)}$ increases in $r$. That is, when a fraction of the data is generated synthetically by the deployed model itself, concentration increases over repeated retraining.
\end{proposition}

\textbf{Agentic overreliance.} Proposition \ref{prop:del_bias} highlights a closed-loop risk in generative recommenders: when users increasingly delegate decisions to AI agents (e.g., autoplay in
streaming or one-click purchase in shopping), the platform’s interaction logs cease to be a passive record of latent demand and instead become an endogenous byproduct of the agent’s own response. These feedback risks are particularly salient now because agentic AI systems reduce the cognitive cost of search and are increasingly gaining user trust in defaults. As the agent becomes more capable, users rely on it more heavily. Yet greater reliance also increases the fraction of agent-shaped interactions that enter the data pipeline.

Because cross-entropy training fits the conditional outcome distribution observed in the logs (Lemma \ref{lem:del_eq2}), retraining on agent-shaped or synthetic data causes the next model to reproduce and amplify the same patterns.  Over time, this shifts probability mass toward a narrow ``style” region of the catalog, thereby suppressing the discovery of alternatives. Proposition \ref{prop:del_bias} therefore shows that the long-run diversity of recommendations is governed by the share of user-organic preference signal relative to agent-delegated interactions. Systems optimized for short-run engagement through heavy automation can degrade the informational content of their data pipelines, producing increasingly homogeneous recommendations and greater vulnerability to shocks (e.g., trend shifts, new entrants, inventory changes) despite a large underlying catalog. The same logic echoes emerging concerns about ``model collapse'' when LLM-style AI agents (beyond recommendation) are repeatedly trained on model-shaped outputs \citep{alemohammad2023self, shumailov2024ai}.

\section{Conclusion}\label{sec:conclusion}

LLM-style and transformer-based agentic AI systems are being deployed rapidly as shopping assistants, streaming copilots, and content navigators. This shift has raised growing concerns, both in research and in practice, regarding the reliability of such agents and the potential for systematic distortions at scale. In this paper, we take a step toward formalizing these concerns through a mechanism-level theoretical analysis of transformer-based generative recommenders. We identify four bias channels that are central in the emerging literature. (i) \emph{Positional bias}: stronger relative positional encodings tilt personalization toward more recent history, trading off responsiveness against stability. (ii) \emph{Popularity amplification}: frequency imbalance in training data can translate into super-linear exposure imbalance through the dot-product-softmax mechanism. (iii) \emph{Latent-driver bias}: when user actions are partly driven by unobserved factors, softmax can amplify small perturbations into large swings in attention. (iv) \emph{Synthetic-data bias}: as users increasingly delegate decisions and platforms retrain on model-shaped synthetic logs, closed-loop feedback can concentrate outputs and steer the system toward a narrower portion of the catalog.

Although our formal analysis is developed in a generative recommender setting, we view this setting as a tractable lens on a broader class of agentic AI systems. Many transformer-based and LLM-style agents on digital platforms operate by repeatedly conditioning on an evolving context and allocating attention across past signals, not only for ID-based systems, but also for language-based or image-based agents. In such systems, the same mechanism-level risks studied here can arise in closely related forms. Position-dependent weighting can distort which evidence is used, consistent with empirical findings on long-context LLMs and agentic shopping systems \citep{liu2024lost, xiao2023efficient, allouah2025your}. Frequency imbalances can be amplified into concentration and homogenization, echoing recent concerns about LLM-style monoculture beyond recommendation \citep{wu2025generative}. Incomplete or noisy
context can induce brittle attention patterns, consistent with emerging concerns about reliability under latent or missing drivers \citep{zhai2023stabilizing}. And once model-shaped outputs begin to enter the future training pipeline, closed-loop retraining can narrow variety over time, in line with the general model-collapse literature \citep{alemohammad2023self, shumailov2024ai}.
We do not claim that our recommender model fully describes all such deployments. Rather, we hope it provides a useful analytical foundation for studying attention-based context weighting and the reliability of endogenous feedback.

A key direction for future work is to connect these mechanism-level bias channels to downstream operational outcomes and to develop governance and control policies that mitigate bias risks. Recent OM/OR work already cautions that as AI systems gain greater autonomy, they should be paired with stronger structure, monitoring, constraints, and escalation logic rather than looser oversight \citep{cohen2026om, dai2025assured}. Moreover, several papers show that good empirical performance can coexist with important fragilities: pretrained transformers for sequential decisions face out-of-distribution and performative-feedback issues \citep{wang2024understanding}; LLM-based newsvendor studies document persistent deviations and bias amplification even when the optimal formula is available \citep{liu2025large, wang2025human}; and synthetic-data approaches in market research rely on the human-AI response gap being sufficiently limited, stable, or estimable to permit debiasing or optimal mixing \citep{wang2024large, yin2026synthetic}. Therefore, high empirical performance alone should not be taken as evidence of reliability, since systems that perform well on average may still be vulnerable to biases over time. Existing work also suggests mitigation directions, such as using LLMs as interfaces that clarify goals while leaving exact computation to OR solvers \citep{duan2025ask, huang2025orlm}, preserving meaningful human oversight and override rather than full automation \citep{wang2025human, baek2026ai}, and treating synthetic data as a calibrated complement rather than a direct substitute for human behavior \citep{zhang2026llm,yin2026synthetic}. Moreover, an important agenda for OM is not only to deploy increasingly powerful AI systems for revenue alone, but also to understand how to use them to achieve greater operational and societal benefits such as customer utility \citep{chen2026utility}. These require developing robust control, monitoring, and intervention mechanisms, whether through human override and escalation or through mechanism-level changes in training, inference, and feedback control.

\bibliographystyle{informs2014} 
\bibliography{ref}

\ECSwitch

\ECHead{\centering Online Appendix to ``Transformer biases"}

\section{Positional Bias: Auxiliary Results and Proofs}

\subsection{Proof of Proposition \ref{prop:position}}\label{pf:prop_position}

By definition,
\begin{align}
    M_S(\alpha) = \sum_{i \in S} A_{l,i}(\alpha), \nonumber
\end{align}
is a function of positional strength $\alpha$, indicating how strong the role relative PE plays when it weighs more. We denote $d_i = b(|\widetilde{l}-i|)$. Differentiating gives
\begin{align}
    \frac{\mathrm{d}M_S(\alpha)}{\mathrm{d}\alpha}
    = & \sum_{i \in S} \sum_{j=1}^T \frac{\partial A_{l,i}}{\partial Z_j} \frac{\partial Z_j(\alpha)}{\partial \alpha} \nonumber \\
    = & \sum_{i \in S} \sum_{j=1}^T A_{l,i}(\alpha)(\delta_{i,j} - A_{l,j}(\alpha)) d_j \label{pf:prop_position_1}\\
    = & \sum_{i \in S} A_{l,i}(\alpha) d_i
     - A_{l,i}(\alpha) \sum^T_{j=1} A_{l,j}(\alpha) d_j \nonumber \\
    = & \sum_{i \in S} A_{l,i}(\alpha)(d_i - \bar{d}(\alpha)), \label{pf:prop_position_2}
\end{align}
where $\delta_{i,j}$ in \eqref{pf:prop_position_1} is the Kronecker delta, and $\bar{d}(\alpha) := \sum^T_{j=1} A_{l,j}(\alpha) d_j$ in \eqref{pf:prop_position_2}.

Recall the index set $\{1, \dots, T\}$ is partitioned into the ``near'' group $S$ and
``far'' group $F := \{1, \dots, T\} \setminus S$. Define
\begin{align}
    \alpha_S(\alpha) &:= \sum_{i \in S} A_{l,i}(\alpha) \in(0,1), \quad \alpha_F(\alpha) := \sum_{i \in F} A_{l,i}(\alpha) = 1-\alpha_S(\alpha), \nonumber \\
    \mu_S(\alpha) &:= \frac{1}{\alpha_S(\alpha)} \sum_{i \in S} A_{l,i}(\alpha) d_i, \quad \mu_F(\alpha) := \frac{1}{\alpha_F(\alpha)} \sum_{i \in F} A_{l,i}(\alpha) d_i. \nonumber
\end{align}
We can rewrite the average bias as
\begin{align}
    \bar d(\alpha) = \sum^T_{j=1} A_{l,j}(\alpha) d_j = \alpha_S(\alpha)\mu_S(\alpha) + \alpha_F(\alpha) \mu_F(\alpha).
\end{align}
Substituting this into \eqref{pf:prop_position_2}, we obtain
\begin{align}
    \mathrm{d}M_S(\alpha)/\mathrm{d}\alpha = & \sum_{i \in S} A_{l,i}(\alpha) d_i - \sum_{i \in S} A_{l,i}(\alpha) \bar{d}(\alpha) \nonumber \\
    = & \alpha_S(\alpha)\mu_S(\alpha) - \alpha_S(\alpha) (\alpha_S(\alpha) \mu_S(\alpha) + \alpha_F(\alpha) \mu_F(\alpha)) \nonumber\\
    = & \alpha_S(\alpha) \alpha_F(\alpha) (\mu_S(\alpha) - \mu_F(\alpha)). \label{pf:prop_position_3}
\end{align}
Thus the sign of $\mathrm{d}M_S(\alpha)/\mathrm{d}\alpha$ is determined by $\mu_S(\alpha) - \mu_F(\alpha)$, since $\alpha_S(\alpha),\alpha_F(\alpha) \geq 0$. By Assumption \ref{asp:RPE_monotone}, for any $i \in S$ and $j \in F$, we have $d_i \geq d_j$. Therefore $\min_{i \in S} d_i \geq \max_{j \in F} d_j$. Hence, 
\begin{align}
  \mu_S(\alpha) = \frac{1}{\alpha_S(\alpha)} \sum_{i \in S} A_{l,i}(\alpha) d_i \geq \min_{i \in S} d_i \geq \max_{j \in F} d_j \geq \frac{1}{\alpha_F(\alpha)} \sum_{j \in F} A_{l,j}(\alpha) d_j = \mu_F(\alpha). \nonumber
\end{align}
Thus, $\mu_S(\alpha) - \mu_F(\alpha) \geq 0$, and $\mathrm{d}M_S(\alpha)/\mathrm{d}\alpha \geq 0$ in this case. 
$\hfill\square$

\subsection{Position patterns of RoPE}\label{pf:RoPE}

\textbf{Rotary positional encoding (RoPE).} Let the cross-attention head dimension be even, $d^{\text{cross}} = 2R$, so that $\bq_l,\bk_j\in\mathbb R^{2R}$. For $r=1,\dots,R$, let $\omega_r>0$ denote the frequency of the $r$-th two-dimensional rotary block, and define the $2\times 2$ rotation matrix
\begin{align}
    R(\varphi) :=
    \begin{bmatrix}
        \cos \varphi & -\sin \varphi\\
        \sin \varphi & \cos \varphi
    \end{bmatrix}.
\end{align}
The block-diagonal rotary operator at position $t$ is $$\mathcal R_t(\theta) := \text{diag}(R(\omega_1\theta t), R(\omega_2\theta t), \ldots, R(\omega_R\theta t)),$$ where $\theta \geq 0$ controls the overall rotary scale, and $0< \omega_1 \leq \cdots \leq \omega_R$ are fixed band multipliers. The RoPE-modified attention logit from decoder position $l$ to history position $j$ is
\begin{align}
    Z^{\text{RoPE}}_{l,j}(\theta) := \frac{(\mathcal R_{\widetilde{l}}(\theta)\bq_l)^\top(\mathcal R_j(\theta)\bk_j)}{\sqrt{d^{\text{cross}}}} = \frac{\bq_l^\top \mathcal R_{\widetilde{l}}(\theta)^\top \mathcal R_j(\theta)\bk_j}{\sqrt{d^{\text{cross}}}}. \label{eq:RoPE_Z}
\end{align}
Thus, RoPE injects positional information by rotating each two-dimensional query/key block by an angle proportional to position, so that the resulting logit depends on the relative phase difference between positions.

\textbf{Reduction to cosine form.} Let $d_j := \widetilde{l} - j$ denote the relative distance from history position $j$ to the current decoder position. Using the orthogonality and group property of planar rotations, we have
\begin{align}
    \mathcal{R}_{\widetilde{l}}(\theta)^\top \mathcal{R}_j(\theta) = \mathcal{R}_{j-\widetilde{l}}(\theta). \nonumber
\end{align}
Therefore, from \eqref{eq:RoPE_Z},
\begin{align}
    Z^{\text{RoPE}}_{l,j}(\theta) =
    \frac{\bq_l^\top \mathcal{R}_{\widetilde{l}}(\theta)^\top \mathcal{R}_j(\theta) \bk_j}{\sqrt{d^{\text{cross}}}} = \frac{\bq_l^\top \mathcal{R}_{j-\widetilde{l}}(\theta) \bk_j}{\sqrt{d^{\text{cross}}}}.
\end{align}
Write the query and key in $R = d^{\mathrm{cross}}/2$ two-dimensional blocks:
\begin{align}
    \bq_l^{(r)}
    :=
    \begin{bmatrix}
         q_{l,2r-1}\\
         q_{l,2r}
    \end{bmatrix},
    \quad
    \bk_j^{(r)}
    :=
    \begin{bmatrix}
         k_{j,2r-1}\\
         k_{j,2r}
    \end{bmatrix},
    \quad r=1, \ldots, R. \nonumber
\end{align}
For each block $r$, define
\begin{align}
    a^{(r)}_{l,j} := \frac{\langle \bq_l^{(r)}, \bk_j^{(r)}\rangle}{\sqrt{d^{\text{cross}}}}, \quad  b^{(r)}_{l,j} := \frac{q_{l,2r-1} k_{j,2r} -q_{l,2r} k_{j,2r-1}}{\sqrt{d^{\text{cross}}}}, \nonumber
\end{align}
and let
\begin{align}
    \kappa^{(r)}_{l,j} := \sqrt{\left(a^{(r)}_{l,j}\right)^2 + \left(b^{(r)}_{l,j}\right)^2}, \quad \psi^{(r)}_{l,j} := \operatorname{atan2}\left(b^{(r)}_{l,j}, a^{(r)}_{l,j}\right). \nonumber
\end{align}
Then the $r$-th two-dimensional block contributes
\begin{align}
    \frac{\bq_l^\top \mathcal{R}_{j-\widetilde{l}}(\theta) \bk_j}{\sqrt{d^{\text{cross}}}} = & \frac{\left(\bq_l^{(r)}\right)^\top R\left(\omega_r \theta(j- \widetilde{l})\right) \bk_j^{(r)}}{\sqrt{d^{\mathrm{cross}}}} \nonumber \\
    = & a^{(r)}_{l,j} \cos(\omega_r \theta d_j) + b^{(r)}_{l,j} \sin(\omega_r\theta d_j) \nonumber\\
    = & \kappa^{(r)}_{l,j}\cos\left(\omega_r\theta d_j-\psi^{(r)}_{l,j}\right). \nonumber
\end{align}
Summing over all $R$ rotary blocks yields the RoPE representation
\begin{align}
    Z^{\text{RoPE}}_{l,j}(\theta) = \sum_{r=1}^{R} \kappa^{(r)}_{l,j} \cos\left(\omega_r\theta d_j -\psi^{(r)}_{l,j}\right). \label{eq:RoPE_Z_decomp}
\end{align}
The attention weight on the history token $j$ is
\begin{align}
    A^{\text{RoPE}}_{l,j}(\theta) := \frac{\exp\left(Z^{\text{RoPE}}_{l,j}(\theta)\right)}{\sum_{j'=1}^T \exp\left(Z^{\text{RoPE}}_{l,j'}(\theta)\right)}. \nonumber
\end{align}

\textbf{General-form positional bias.}
Without further restrictions, which position is favored under RoPE depends not only on
relative distance, but also on the block-specific amplitudes $\kappa^{(r)}_{l,j}$ and phases $\psi^{(r)}_{l,j}$. To isolate this
effect, fix one decoder cross-attention row $l$ and consider two history positions
$i,j \in [T]$ with identical key vectors, $\bk_i=\bk_j$, but different relative distances to the current decoder position, i.e., $d_i\neq d_j$. Then for each rotary block $r$,
\begin{align}
    \kappa^{(r)}_{l,i} = \kappa^{(r)}_{l,j} =: \kappa_r, \quad \psi^{(r)}_{l,i} = \psi^{(r)}_{l,j} =: \psi_r. \nonumber
\end{align}
Therefore,
\begin{align}
    Z^{\mathrm{RoPE}}_{l,i}(\theta) - Z^{\mathrm{RoPE}}_{l,j}(\theta) = \sum_{r=1}^{R} \kappa_r \left[\cos(\omega_r \theta d_i - \psi_r)-\cos(\omega_r \theta d_j - \psi_r) \right]. \nonumber
\end{align}
This difference need not be zero even though the two tokens are identical in content. However, its sign is not determined by distance alone, since it depends on the phase-adjusted oscillatory terms across rotary blocks. Thus, in its most general form, RoPE can induce positional bias, but that bias need not take the form of a specific recency or primacy preference.

\textbf{Small-$\theta$ regime.} The exact RoPE relevance score $Z^{\mathrm{RoPE}}_{l,j}(\theta)$ entangles the relative-distance term
$\omega_r\theta d_j$ with content-specific amplitudes $\kappa^{(r)}_{l,j}$ and phases
$\psi^{(r)}_{l,j}$, which makes it difficult to isolate a clean positional factor. To study the positional effect itself, we therefore consider the reduced-form coherent-band
approximation
\begin{align}
    \widetilde Z^{\text{RoPE}}_{l,j}(\theta) := \sum_{r=1}^{R} \kappa^{(r)}_{l} \cos(\omega_r\theta d_j), \quad j\in[T], \nonumber
\end{align}
where $\kappa^{(r)}_{l} \geq 0$ and $d_j:=\widetilde{l} - j$. This should be viewed as a tractable, position-isolating approximation rather than a claim about the most general trained RoPE model. As in the RPE case, fix a near set $S \subseteq [T]$, let $F:=[T] \setminus S$ where
\begin{align}
    d_i \leq d_j \quad \text{for all } i\in S,\ j \in F. \label{eq:RoPE_near_far}
\end{align}
Define
\begin{align}
    M_S(\theta) := \sum_{i\in S}\widetilde A^{\text{RoPE}}_{l,i}(\theta), \nonumber\
\end{align}
where $\widetilde A^{\text{RoPE}}_{l,i}(\theta)$ is the softmax weight induced by the logits $\widetilde Z^{\text{RoPE}}_{l,j}(\theta)$.

\begin{lemma}[Rotation angle tilts attention toward near history] \label{lem:RoPE_small}
Define $d_{\max}:=\max_{j\in[T]} d_j$ and  $\omega_{\max}:=\max_{1\le r\le R}\omega_r$. Suppose
\begin{align}
     \theta \le \frac{\pi}{2 \omega_{\max} d_{\max}}. \label{eq:RoPE_small_condition}
\end{align}
Then
\begin{align}
    \frac{d}{d\theta} M_S(\theta) \geq 0.
\end{align}
That is, within this coherent-band small-angle regime, increasing the rotary scale $\theta$
shifts attention mass toward the near-history window.
\end{lemma}

\begin{proof}{Proof of Lemma \ref{lem:RoPE_small}}
We omit the superscript RoPE in the following proof for notational convenience. Similar with the RPE case, differentiating $M_S(\theta)$ gives
\begin{align}
    \frac{d}{d\theta}M_S(\theta) & = \sum_{i\in S} \sum_{j=1}^T \frac{\partial \widetilde{A}_{l,i}}{\partial \widetilde{Z}_{l,j}} \frac{\partial \widetilde{Z}_{l,j}(\theta)}{\partial \theta} = \sum_{i\in S} \sum_{j=1}^T
    \widetilde{A}_{l,i}(\theta) \left(\delta_{i,j}-\widetilde{A}_{l,j}(\theta)\right) g_j(\theta), \nonumber
\end{align}
where
\begin{align}
    g_j(\theta) := \frac{\partial \widetilde{Z}_{l,j}(\theta)}{\partial \theta} = -\sum_{r=1}^{R} \kappa^{(r)}_{l} \omega_r d_j \sin(\omega_r\theta d_j).
    \nonumber
\end{align}
Thus,
\begin{align}
    \frac{d}{d\theta}M_S(\theta) & = \sum_{i\in S} \widetilde{A}_{l,i}(\theta) g_i(\theta) - \sum_{i\in S} \widetilde{A}_{l,i}(\theta) \sum_{j=1}^T \widetilde{A}_{l,j}(\theta) g_j(\theta) \nonumber \\
    & = \sum_{i\in S}\widetilde{A}_{l,i}(\theta) \left(g_i(\theta)-\bar{g}(\theta)\right),
    \label{pf:rope_small_1}
\end{align}
where
\begin{align}
    \bar{g}(\theta) := \sum_{j=1}^T \widetilde{A}_{l,j}(\theta) g_j(\theta). \nonumber
\end{align}
Define
\begin{align}
    \alpha_S(\theta) &:= \sum_{i\in S} \widetilde{A}_{l,i}(\theta) \in (0,1), \quad \alpha_F(\theta) := 1 - \alpha_S(\theta), \nonumber \\
    \mu_S^g(\theta) &:= \frac{1}{\alpha_S(\theta)} \sum_{i\in S} \widetilde{A}_{l,i}(\theta) g_i(\theta), \quad \mu_F^g(\theta) := \frac{1}{\alpha_F(\theta)} \sum_{j\in F} \widetilde{A}_{l,j}(\theta) g_j(\theta). \nonumber
\end{align}
Then
\begin{align}
    \bar{g}(\theta) = \alpha_S(\theta) \mu_S^g(\theta) + \alpha_F(\theta) \mu_F^g(\theta), \nonumber
\end{align}
and substituting this into \eqref{pf:rope_small_1} yields
\begin{align}
    \frac{d}{d\theta} M_S(\theta) = \alpha_S(\theta) \alpha_F(\theta) \left(\mu_S^g(\theta)-\mu_F^g(\theta)\right). \label{pf:rope_small_2}
\end{align}
Hence the sign of $\frac{d}{d\theta}M_S(\theta)$ is determined by $\mu_S^g(\theta) -\mu_F^g(\theta)$.

It remains to show that $g_i(\theta)\ge g_j(\theta)$ whenever $i\in S$ and $j\in F$. Under \eqref{eq:RoPE_small_condition}, for every $r$ and every $d\in[0,d_{\max}]$,
\begin{align}
    0 \leq \omega_r \theta d \le \omega_{\max}\theta d_{\max}\le \frac{\pi}{2}, \nonumber
\end{align}
so both $\sin(\omega_r\theta d)$ and $\cos(\omega_r\theta d)$ are nonnegative. Therefore, for each $r=1, \ldots, R$,
\begin{align}
    \frac{\partial}{\partial d} \left(d \sin(\omega_r\theta d)\right) = \sin(\omega_r\theta d) + \omega_r\theta d\cos(\omega_r\theta d) \geq 0, \nonumber
\end{align}
which is nondecreasing in $d$ on $[0,d_{\max}]$. Since for $i\in S$ and $j\in F$, $d_i\le d_j$,
\begin{align}
    d_i \sin(\omega_r\theta d_i) \leq d_j \sin(\omega_r\theta d_j)
\end{align}
for all $r = 1, \ldots, R$. Multiplying by $-\kappa^{(r)}_l\omega_r\le 0$ and summing over $r$ yields $g_i(\theta) \geq g_j(\theta)$. Therefore $\mu_S^g(\theta)\ge \mu_F^g(\theta)$ and the claim is proved.
$\hfill\square$
\end{proof}

\begin{remark}[Relation to the standard RoPE parameterization] 
Under the standard RoPE parameterization, the $r$-th two-dimensional block rotates at per-position angular frequency $\theta_r^{\text{std}} = 10000^{-2(r-1)/d^{\text{cross}}}$ for all $r=1,\ldots,R$ up to minor indexing conventions. In our notation, the angle of block $r$ at position $t$ is $\omega_r \theta t$, so the exact correspondence is $\omega_r\theta = \theta_r^{\mathrm{std}} = 10000^{-2(r-1)/d^{\mathrm{cross}}}$. Hence the relevant relative phase difference at history distance $d_j$ is $(\omega_r\theta)d_j = 10000^{-2(r-1)/d^{\mathrm{cross}}} d_j$. Accordingly, the small-angle condition in Lemma \ref{lem:RoPE_small}, $\theta \leq \pi/(2\omega_{\max}d_{\max})$, is equivalent to requiring that every \emph{active} rotary block satisfy $(\omega_r\theta)d_j \leq \pi/2$. Under the standard parameterization, this is $10000^{-2(r-1)/d^{\mathrm{cross}}} \cdot d_{\max} \leq \pi/2$ for the active band under study. This interpretation is practically meaningful when the row-wise positional effect is dominated by slowly rotating, low-frequency RoPE dimensions. Indeed, prior analysis of RoPE emphasizes that transformers often predominantly utilize slowly rotating feature dimensions \citep{barbero2025round}, and notes that the base rotational angle can be as small as approximately $1/10000$ per token (allowing token length $d_{\max} \approx \pi/(2 \times 10^{-4})$), leading to gradual decay effects rather than rapid oscillation. In that regime, the small-angle condition can hold over a substantial range of distances. Therefore, our assumption should be interpreted as a sufficiently small-angle condition for an active low-frequency band. 
\end{remark}

\section{Popularity bias: Proofs}\label{sec:pf_pop_bias}

We first analyze how the evolving trajectory of $\bmu_h^{(n)}$ depends on token popularity of the training data (forthcoming Lemma \ref{lem:pop_mu}). Then, we evaluate the model bias by feeding in an arbitrary test query $\bq$ (Proposition \ref{prop:pop_bias}).

\subsection{How SGD Trajectory Evolves}

\begin{lemma}\label{lem:pop_mu}
For an early training horizon $N$, fix an SGD step $n \leq N-1$. Define
\begin{align}
    \bw_h := \mathbb{E}\left[\bq^{(n)} \mid  X^{(n)}=h \right], \quad \bar{\bw} := \sum_{h' \in \HHH} p_{h'} \bw_{h'} \nonumber
\end{align}
to be the expected query direction of token $\bh$ and the average query direction for all tokens. Then, the expected one-step SGD update satisfies
\begin{align}
    \mathbb{E}\left[\bmu_h^{(n+1)}-\bmu_h^{(n)}\right] = \frac{\eta p_h}{\sqrt{d}}(\bw_h - \bar{\bw}) - \frac{\eta p_h}{d} \mathbb{E}\left[\bq^{(n)}(\bq^{(n)})^\top\right] \delta_h^{(n)} + O\left(\frac{\eta p_h\epsilon^2}{\sqrt{d}}\E\|\bq^{(n)}\|\right), \nonumber
\end{align}
where $\delta^{(n)}_h = \bmu_h^{(n)} - \sum_{\bh \in \HHH} p_h \bmu_h^{(n)}$ remains small during early training, and $\epsilon$ is a higher order term of $\delta^{(n)}$.
\end{lemma}

\begin{proof}{Proof of Lemma \ref{lem:pop_mu}}
\textbf{Step 1: Gradient formula.} Fix a token $\bh \in \HHH$, we have
\begin{align}
    \nabla_{\bmu_h}\ell_{\theta}(\bq,X) = \sum^{t \times L}_{j=1} \frac{\partial \ell}{\partial Z_j} \cdot \frac{\partial Z_j}{\partial \bmu_h} = \sum^{t \times L}_{j=1} \sum ^{t \times L}_{i=1} \frac{\partial \ell}{\partial A_i} \cdot \frac{\partial A_i}{\partial Z_j} \cdot \frac{\partial Z_j}{\partial \bmu_h}. \label{pf:pop_l}
\end{align}
We now derive each term in \eqref{pf:pop_l}. Since $M_h(\bq) = \sum_{j:\,\bh_j=\bh} A_j$ (we shorthand $M_h(\bq)$ as $M_h$), we have
\begin{align}
    \frac{\partial \ell}{\partial A_i} = -\frac{1}{M_X} \cdot \frac{\partial M_X}{\partial A_i} = -\frac{1}{M_X} \cdot \mathbf{1}\{i \in S_X\}. \label{pf:pop_l_1}
\end{align}
The softmax Jacobian gives
\begin{align}
    \frac{\partial A_i}{\partial Z_j} = A_i(\mathbf{1}\{i = j\} - A_j). \label{pf:pop_l_2}
\end{align}
We also have
\begin{align}
    \frac{\partial Z_j}{\partial \bmu_h} = \frac{\partial \left(\langle \bq, \bmu_h \rangle/\sqrt{d}\right)}{\partial \bmu_h} = \frac{\bq}{\sqrt{d}} \cdot \mathbf{1}\{j \in S_h\}. \label{pf:pop_l_3}
\end{align}
Putting together \eqref{pf:pop_l_1}, \eqref{pf:pop_l_2} and \eqref{pf:pop_l_3}, we have
\begin{align}
    \nabla_{\bmu_h}\ell = & \sum^{t \times L}_{j=1} \sum ^{t \times L}_{i=1} -\frac{1}{M_X} \cdot \mathbf{1}\{i \in S_X\} \cdot A_i(\mathbf{1}\{i = j\} - A_j) \cdot \frac{\bq}{\sqrt{d}} \cdot \mathbf{1}\{j \in S_h\} \nonumber \\
    = & \sum^{t \times L}_{j=1} -\frac{1}{M_X} \cdot \left(\sum_{i \in S_X} A_i(\mathbf{1}\{i = j\} - A_j)\right) \cdot \frac{\bq}{\sqrt{d}} \cdot \mathbf{1}\{j \in S_h\} \nonumber \\
    = & \sum^{t \times L}_{j=1} -\frac{1}{M_X} \cdot \left(A_j \cdot \mathbf{1}\{j \in S_X\} - A_j \cdot M_X\right) \cdot \frac{\bq}{\sqrt{d}} \cdot \mathbf{1}\{j \in S_h\} \label{pf:pop_l_4} \\
    = & \sum^{t \times L}_{j=1} \left(A_j - \frac{A_j}{M_X} \mathbf{1}\{j \in S_X\}\right) \cdot \frac{\bq}{\sqrt{d}} \cdot \mathbf{1}\{j \in S_h\} \nonumber \\
    = & \frac{\bq}{\sqrt{d}} \left(\sum_{j \in S_h} A_j - \frac{1}{M_X} \sum_{j \in S_h \cap S_X} A_j \right) \nonumber \\
    = & \frac{\bq}{\sqrt{d}} \left(M_h - \frac{M_h}{M_X} \cdot \mathbf{1}\{X = \bh\} \right). \label{pf:pop_l_5}
\end{align}
Equation \eqref{pf:pop_l_4} holds since $\sum_{i \in S_X} A_i \cdot \mathbf{1}\{i = j\} = A_j \cdot \mathbf{1}\{j \in S_X\}$ and $\sum_{i \in S_X} A_i  A_j = M_X A_j$. Equation \eqref{pf:pop_l_5} stands because
\[
\sum_{j \in S_h \cap S_X} A_j =
\begin{cases}
    M_h &\text{if } X = \bh \\
    0 &\text{if } X \neq \bh
\end{cases}
\;= M_h \cdot \mathbf{1}\{X = \bh\}.
\]
Hence,
\begin{align}
    \nabla_{\bmu_h}\ell_{\theta}(\bq,X) = \frac{1}{\sqrt{d}} \left(M_h(\bq) - \mathbf{1}\{X = \bh\} \right) \bq. \label{pf:pop_l_dev}
\end{align}

\textbf{Step 2: SGD early-stage amplification.} By the SGD update rule \eqref{eq:pop_sgd_update} and \eqref{pf:pop_l_dev},
\begin{align}
    \bmu_h^{(n+1)} - \bmu_h^{(n)} =  - \frac{\eta}{\sqrt{d}} \left(M^{(n)}_h(\bq^{(n)}) - \mathbf{1}\{\bh = X^{(n)}\} \right) \bq^{(n)} \label{pf:pop_step}.
\end{align}
We define $\bar{\bmu}^{(n)} := \sum_{\bh \in \HHH} p_h \bmu^{(n)}_h$ and $\delta^{(n)}_h = \bmu^{(n)}_h - \bar{\bmu}^{(n)}$. We write $M_h^{(n)}(\bq; \bmu_h^{(n)})$ (see \eqref{eq:pop_M_n}) as
\begin{align}
    M_h^{(n)}(\bq; \bmu_h^{(n)}) = & \frac{|S_h|\exp\left(\langle \bq, \bmu_h^{(n)}\rangle / \sqrt{d}\right)}{\sum_{\bh' \in \mathcal{H}} |S_{h'}|\exp\left(\langle \bq, \bmu_{h'}^{(n)}\rangle / \sqrt{d}\right)} \nonumber \\
    = & \frac{|S_h|\exp\left(\langle \bq, \bar{\bmu}^{(n)} + \delta^{(n)}_h \rangle / \sqrt{d}\right)}{\sum_{\bh' \in \mathcal{H}} |S_{h'}|\exp\left(\langle \bq, \bar{\bmu}_{h'} + \delta^{(n)}_{h'}\rangle / \sqrt{d}\right)} \nonumber \\
     = & \frac{|S_h|\exp\left(\langle \bq, \delta^{(n)}_h \rangle / \sqrt{d}\right)}{\sum_{\bh' \in \mathcal{H}} |S_{h'}|\exp\left(\langle \bq, \delta^{(n)}_{h'}\rangle / \sqrt{d}\right)}. \nonumber
\end{align}
At the initial of training, $\bmu^{(n)}_h$ are not distinguished much between tokens, and $\delta^{(n)}_h \approx 0$ for all $\bh \in \HHH$. Therefore, for early SGD steps $n = 0, 1, \ldots, N-1$, $\delta^{(n)}_h$ satisfies
\begin{align}
    \Bigg| \frac{\langle \bq, \delta^{(n)}_h \rangle}{\sqrt{d}} \Bigg| \leq \epsilon \quad \forall \bq, \bh,
\end{align}
where $\epsilon$ is a small constant. Shorthanding $z^{(n)}_h(\bq): = \langle \bq, \delta^{(n)}_h \rangle / \sqrt{d}$, we have
\begin{align}
    M_h^{(n)}(\bq; \bmu_h^{(n)}) = & \frac{|S_h| e^{z^{(n)}_h(\bq)}}{\sum_{\bh' \in \mathcal{H}} |S_{h'}|e^{z^{(n)}_{h'}(\bq)}} \nonumber \\
    = & \frac{|S_h|\left(1 + z^{(n)}_h(\bq) + O\left(\epsilon^2\right)\right)}{\sum_{\bh' \in \mathcal{H}} |S_{h'}|\left(1 + z^{(n)}_{h'}(\bq) + O\left(\epsilon^2\right)\right)} \label{pf:pop_Taylor} \\
    = & \frac{p_h \left(1 + z^{(n)}_h(\bq) + O\left( \epsilon^2 \right)\right)}{\sum_{\bh' \in \mathcal{H}} p_{h'} \left(1 + z^{(n)}_{h'}(\bq) + O\left(\epsilon^2\right)\right)} \label{pf:pop_Taylor_1} \\
    = & p_h \left(1 + z^{(n)}_h(\bq) + O\left( \epsilon^2\right)\right). \label{pf:pop_Taylor_2}
\end{align}
Equation \eqref{pf:pop_Taylor} follows from using Taylor expansion around $\delta^{(n)}_h = 0$. Equation \eqref{pf:pop_Taylor_1} holds since $p_h = |S_h|/T$ for all $\bh \in \HHH$. Equation \eqref{pf:pop_Taylor_2} follows because
\begin{align}
    & \sum_{\bh' \in \mathcal{H}} p_{h'} = 1, \nonumber \\
    &\sum_{\bh' \in \HHH} p_{h'} z^{(n)}_{h'}(\bq) = \sum_{\bh' \in \HHH} p_{h'} \frac{\langle \bq, \delta^{(n)}_{h'} \rangle}{\sqrt{d}} = \frac{\langle \bq, \sum_{\bh' \in \HHH} p_{h'} (\bmu^{(n)}_{h'} - \bar{\bmu}^{(n)}) \rangle}{\sqrt{d}} = 0. \nonumber
\end{align}
Hence, the denominator of \eqref{pf:pop_Taylor_1} simplifies to $1 + O(\epsilon^2)$. We define
\begin{align}
    \bw_h := \mathbb{E}\left[\bq^{(n)}\mid X^{(n)}=h\right], \nonumber
\end{align}
which can be interpreted as the direction the query context signal given token $\bh$. We further define
\begin{align}
    \bar{\bw} := \mathbb{E}\left[\bq^{(n)}\right]= \sum_{\bh \in \HHH} \mathbb{E}\left[\bq^{(n)} \Big| \bh = X^{(n)}\right] \Pr(\bh = X^{(n)}) =\sum_{\bh \in \HHH} p_h \bw_h \nonumber
\end{align}
to be the average direction of the query context. The expectations are over the training randomness $(\bq^{(n)},X^{(n)})$. Then,
\begin{align}
    \mathbb{E}\left[\mathbf{1}\{X^{(n)} = \bh\} \bq^{(n)}\right] = \Pr\left(X^{(n)} = \bh\right) \mathbb{E}\left[\bq^{(n)} \big| \bh = X^{(n)}\right] = p_h \bw_h, \label{pf:pop_exp_1}
\end{align}
and
\begin{align}
    \mathbb{E}\left[M^{(n)}_h(\bq^{(n)}) \bq^{(n)} \right] = & \mathbb{E}\left[p_h \left(1 + z^{(n)}_h(\bq^{(n)}) + O\left( \epsilon^2\right)\right) \bq^{(n)} \right] \nonumber \\
    = & p_h \mathbb{E}\left[\bq^{(n)}\right] + p_h \mathbb{E}\left[z_h^{(n)}(\bq^{(n)})\,\bq^{(n)}\right] + O\left(p_h\epsilon^2 \mathbb{E}\left\|\bq^{(n)}\right\|\right) \nonumber \\
    = & p_h \bar{\bw} + p_h \mathbb{E}\left[\frac{\langle \bq^{(n)}, \delta_h^{(n)} \rangle}{\sqrt d}\, \bq^{(n)}
    \right] + O\left(p_h \epsilon^2 \mathbb{E}\|\bq^{(n)}\|\right) \nonumber \\
    = & p_h \bar{\bw} + \frac{p_h}{\sqrt{d}} \mathbb{E}\left[\bq^{(n)}(\bq^{(n)})^\top\right] \delta_h^{(n)} + O\left(p_h \epsilon^2 \mathbb{E}\|\bq^{(n)}\|\right). \label{pf:pop_exp_2}
\end{align}
Combining \eqref{pf:pop_exp_1} and \eqref{pf:pop_exp_2}, we have
\begin{align}
    \mathbb{E}\left[\bmu_h^{(n+1)} - \bmu_h^{(n)}\right] = & \frac{\eta}{\sqrt{d}} \left(\mathbb{E}\left[\mathbf{1}\{X^{(n)} = \bh\} \bq^{(n)}\right] - \mathbb{E}\left[M^{(n)}_h(\bq^{(n)}) \bq^{(n)} \right] \right) \nonumber \\
    = & \frac{\eta p_h}{\sqrt{d}}(\bw_h - \bar{\bw}) - \frac{\eta p_h}{d} \mathbb{E}\left[\bq^{(n)}(\bq^{(n)})^\top\right] \delta_h^{(n)} + O\left(\frac{\eta p_h\epsilon^2}{\sqrt{d}}\E\|\bq^{(n)}\|\right). \label{pf:pop_exp_3}
\end{align}
Thus, in the early stage when $\delta_h^{(n)}$ remains small, the leading update direction is
\begin{align}
    \mathbb{E}\left[\bmu_h^{(n+1)} - \bmu_h^{(n)}\right] = \frac{\eta p_h}{\sqrt d}(\bw_h-\bar{\bw}) + \text{higher-order terms}. \label{pf:pop_update}
\end{align}
Hence, as long as the early-stage centered deviation $\delta_h^{(n)}$ and Taylor remainder remain sufficiently small, the token $\bh$'s learning speed is proportional to $p_h$. Therefore, among tokens with comparable query-context magnitude, more frequent tokens receive greater early-stage directional reinforcement.
$\hfill\square$
\end{proof}

\begin{remark}{\textbf{(Why early-stage amplification matters at the final stage training).}}
\label{rem:early_window}
Lemma \ref{lem:pop_mu} says that SGD pulls $\bmu_h$ in the direction of $\bw_h - \bar{\bw}$ at a rate multiplied by $p_h$. It hinges on an early training window in which the model's own attention weights are not yet too sharply concentrated. In particular, within the first $N$ updates, tokens with higher frequency $p_h$ accrue a greater representation advantage. 

This early-stage result is meaningful for the final trained iterate used in practice because training is typically stopped based on the improvement in the average loss $\ell_{X^{(n)}}$ over all tokens, rather than on verifying that each token-specific loss has converged. In our setting, the empirical cross-entropy is frequency-weighted, so residual errors on low-frequency tokens contribute only weakly to the observed improvement in the overall objective. Consequently, once frequent tokens are fitted well enough that the average one-step loss decrease becomes small, optimization may stop even though rare tokens remain under-trained. In this sense, an early popularity-aligned advantage need not dominate the entire trajectory to matter at the final model returned by training. It suffices that the advantage is created before the stopping criterion is triggered and is not fully corrected before aggregate improvement falls below the practical convergence threshold. 

Lemma \ref{lem:pop_mu} identifies a concrete bias-creation window. It does not claim that the same approximation remains valid for all later iterations. Rather, it shows that dot-product-softmax attention can generate a frequency-aligned asymmetry early in training, and a standard average-loss stopping rule may then preserve this asymmetry in the final iterate by terminating before low-frequency tokens are fully learned. This is consistent with prior empirical evidence that early frequency-driven advantages can create persistent popularity bias \citep{chen2023adap} and early deficits can leave a lasting effect on the final trained model \citep{kleinmancritical}.
\end{remark}

\subsection{Proof of Proposition \ref{prop:pop_bias}}

As noted in Remark \ref{rem:early_window}, under the SGD stopping rule and empirical evidence, the early-stage directional bias accumulates and is often not fully erased. In this case, the accumulated early-stage advantage persists in the sense that at the stopping step $\widehat{N}$,
\begin{align}
    \mathbb{E}\left[\bmu_h^{(\widehat{N})}\right] = \bmu_h^{(0)} + \frac{\eta \widehat{N} p_h}{\sqrt{d}}(\bw_h - \bar{\bw}) + \br^{(\widehat{N})}_h, \label{pf:pop_persist}
\end{align}
where 
\begin{align}
    \left\|\br^{(\widehat{N})}_h\right\| \leq \xi^{(\widehat{N})} \label{pf:pop_r}
\end{align}
for some small $\xi^{(\widehat{N})}$. 

After obtaining the trained $\{\bmu^{(n)}\}^{\widehat{N}}_{n=1}$, for any fixed query $\bq$, recall that by definition
\begin{align}
    M_h^{(\widehat{N})}(\bq) = \frac{p_h \exp\left(\langle \bq, \bmu_h^{(\widehat{N})}\rangle / \sqrt{d}\right)}{\sum_{h' \in \HHH} p_{h'} \exp\left(\langle \bq, \bmu_{h'}^{(\widehat{N})}\rangle / \sqrt{d}\right)}. \nonumber
\end{align}
Hence the denominator cancels in the ratio, and
\begin{align}
    \frac{M_h^{(\widehat{N})}(\bq)}{M_{h'}^{(\widehat{N})}(\bq)} = \frac{p_h}{p_{h'}} \cdot
    \exp\left(\frac{\langle \bq, \bmu_h^{(\widehat{N})} -\bmu_{h'}^{(\widehat{N})}\rangle}{\sqrt{d}}\right). \label{pf:pop_M_ratio}
\end{align}
Therefore, taking expectation over all training randomness, by Jensen's inequality, 
\begin{align}
    \text{AR}(\bh,\bh') = \mathbb{E}\left[\exp\left( \frac{\left\langle \bq, \bmu_h^{(\widehat{N})} -\bmu_{h'}^{(\widehat{N})}\right\rangle}{\sqrt d}\right) \right] \geq \exp\left(\frac{1}{\sqrt{d}} \left\langle \bq, \mathbb{E}\left[\bmu_h^{(\widehat{N})} -\bmu_{h'}^{(\widehat{N})}\right]\right\rangle \right). \label{pf:pop_AR_1}
\end{align}
Using \eqref{pf:pop_persist},
\begin{align}
    \text{AR}(\bh,\bh') \geq  \exp\left(\frac{\eta \widehat{N}}{d} \left\langle \bq, p_h(\bw_h - \bar{\bw}) - p_{h'} (\bw_{h'} - \bar{\bw})\right\rangle + \frac{1}{\sqrt{d}} \left\langle \bq, \br_h^{(\widehat{N})} -\br_{h'}^{(\widehat{N})} \right\rangle\right). \label{pf:pop_AR_2}
\end{align}
Since $\|\br^{(\widehat{N})}\| \leq \xi^{(\widehat{N})}$, by Cauchy-Schwarz, the last term in \ref{pf:pop_AR_2} satisfies
\begin{align}
    \left|\frac{1}{\sqrt{d}} \left\langle \bq, \br_h^{(\widehat{N})} -\br_{h'}^{(\widehat{N})} \right\rangle\right| \leq \frac{\|\bq\|}{\sqrt{d}} \left(\|\br_h^{(\widehat{N})}\| + \|\br_{h'}^{(\widehat{N})}\|\right) \leq \frac{2\xi^{(\widehat{N})}\|\bq\|}{\sqrt{d}}. \nonumber
\end{align}
Therefore, 
\begin{align}
    \text{AR}(\bh,\bh') \geq\exp\left(\frac{\eta \widehat{N}}{d} \left\langle \bq, p_h(\bw_h - \bar{\bw}) - p_{h'} (\bw_{h'} - \bar{\bw})\right\rangle - \frac{2\xi^{(\widehat{N})}\|\bq\|}{\sqrt{d}}\right). \label{pf:pop_M_ratio_3}
\end{align}
which proves the claim.
$\hfill\square$

\section{Latent Driver Bias: Auxiliary Results and Proofs}\label{sec:pf_cloning}

By the definition of softmax,
\begin{align}
    p_j = \frac{e^{Z_j}}{\sum_{j'=1}^T e^{Z_{j'}}}, \quad p_k = \frac{e^{Z_k}}{\sum_{j'=1}^T e^{Z_{j'}}}.
\end{align}
Since both denominators are identical and strictly positive,
\begin{align}
    \frac{p_j}{p_k} = \frac{e^{Z_j}}{e^{Z_k}} = e^{Z_j-Z_k}. \nonumber
\end{align}
Note that $Z_j - Z_k$ is a linear function of the Gaussian vector $Z$, hence is Gaussian. Since $Z \sim \mathcal{N}(0,\Sigma)$, $Z_j-Z_k = (\boldsymbol{e}_j - \boldsymbol{e}_k)^\top Z$ where $\boldsymbol{e}_j, \boldsymbol{e}_k$ are the standard basis vectors. Therefore,
\begin{align}
    & Z_j - Z_k \sim \mathcal{N}\left(0, (\boldsymbol{e}_j - \boldsymbol{e}_k)^\top \Sigma (\boldsymbol{e}_j - \boldsymbol{e}_k)\right) \nonumber \\
    \Leftrightarrow & Z_j - Z_k \sim \mathcal{N}\left(\Sigma_{jj}+\Sigma_{kk}-2\Sigma_{jk}\right). \nonumber
\end{align}
Under Assumption \ref{asp:clone_noise}, $\Sigma=\text{diag}(\sigma_1^2,\dots,\sigma_T^2)$, and therefore
\begin{align}
    \log\left(\frac{p_j}{p_k}\right) \sim \mathcal{N}(0, \sigma_j^2 + \sigma_k^2). \nonumber
\end{align}
Since $p_j/p_k$ is lognormal with log-mean $0$ and log-variance $\sigma_j^2+\sigma_k^2$,
\begin{align}
    \mathbb{E}\left[\frac{p_j}{p_k}\right] = \exp\left(\frac{\sigma_j^2+\sigma_k^2}{2}\right). \nonumber
\end{align}
This is strictly increasing and exponential in the pairwise variance. For any $c>1$,
\begin{align}
    \Pr\left(\frac{p_j}{p_k} \geq c\right) = \Pr\left(\log\left(\frac{p_j}{p_k}\right) \geq \log c\right) = 1- \Phi\left(\frac{\log c}{\sqrt{\sigma_j^2+\sigma_k^2}}\right). \nonumber
\end{align}
The probability is strictly increasing in $\sigma_j^2 + \sigma_k^2$.
$\hfill\square$

\section{Synthetic Data Bias: Auxiliary Results and Proofs}\label{sec:pf_del}
Before diving into the proof of Proposition \ref{prop:del_bias} in section \ref{pf:prop_del}, we introduce the detailed formulation of $\bp_j$ in section \ref{pf:del_onehot} and show that the cross entropy sets the context-conditional distribution to the empirical conditional frequencies in section \ref{pf:del_empirical}. 

\subsection{One-Hot Embedding}\label{pf:del_onehot}
For an event with history encoding $H_t$ and prefix $\by_{<l}$, letting $\widetilde{\by}_l$ be the decoder state after masked self-attention, we define the cross-attention output
\begin{align}
    \br_{t,l}^{(r)}:= \sum_{j=1}^{t \times L}  \text{softmax}\left(\frac{\langle \widetilde{\by}_l W_{Q}^{(m),\text{cross}},\, \bh^{\text{enc}}_j W_{K}^{(m),\text{cross}}\rangle}{\sqrt{d^{\text{cross}}}}\right)_j \cdot \left(\bh^{\text{enc}}_j W_{V}^{(m),\text{cross}}\right),
\end{align}
with $\br_{t,l}:=\br_{t,l}^{(M)}$ after stacking layers and FFNs. We write $\br_{t,l}=R_\theta(H_t, \by_{<l})$ as shorthand, where $\theta$ collects all decoder parameters, including $(W_Q^{\text{cross}}, W_K^{\text{cross}},W_V^{\text{cross}})$. Recall our loss objective
\begin{align}
    \LLL_{\text{SID}} = -\sum_{H_t, Y} \sum^L_{l=1} \log\left(p_{\theta}(c_l|H_t, c_{<l})\right). \nonumber
\end{align}
During training, each pair $(H_t, Y)$ contributes $L$ supervised prediction events. Index all such events by $n=1,\dots,N$ (so $N = \sum_{(H_t,Y)}L$), where event $n$ corresponds to some $(t,l)$ and has label $\by_n$ and context $\br_n :=R_\theta \left(H_t, \by_{<l}\right)$. Substituting this into $\LLL_{\text{SID}}$ shows that the original objective is precisely the cross-entropy of a multiclass classifier evaluated on the (learned) features $\{\br_n\}_{n=1}^N$:
\begin{align}
    \LLL_{\text{SID}} = -\sum_{n=1}^N \log\left(p_{\theta}(\by_n|\br_n)\right) = -\sum_{n=1}^N \log \left(\text{softmax} (W_{\text{out}} \br_n + \bb)_{\by_n}\right). \nonumber
\end{align}
We next consider a model simplification following \citet{seddik2024bad}. Since
\begin{itemize}
    \item the token alphabet $\HHH, \YYY$ is finite,
    \item the semantic ID length $L$ is finite,
    \item during model training, the interaction history is truncated to a finite window (usually the last $T$ tokens),
\end{itemize}
the context space is finite. Hence, we can define a clustering/quantization map $g: \mathbb{R}^d \rightarrow [c]$ that projects each context $\br_n$ to a finite set of $c$ one-hot embeddings:
\begin{align}
    \bx_n :=\boldsymbol{e}_{g(\br_n)} \in \{\boldsymbol{e}_1, \ldots, \boldsymbol{e}_c\} \subset \mathbb{R}^c. \nonumber
\end{align}
Each element of the one-hot embedding set $\boldsymbol{e}_i$ represents one possible state of $\br_n$. Similarly, we can project the possible next token $\by_n$ to a finite set of $s$ one-hot embeddings:
\begin{align}
    \bz_n \in \{\boldsymbol{e}_1, \ldots, \boldsymbol{e}_s\} \subset \mathbb{R}^s. \nonumber
\end{align}
In this case, for each round of training, it can be seen that we have $N$ samples of contexts and next-token pairs $\{(\bx_n, \bz_n)\}^N_{n=1}$ in each retraining round. 

During training, we focus on a surrogate model of 
\begin{align}
    \widehat{\LLL}_{\text{SID}} := - \frac{1}{N} \sum_{n=1}^N \log p(\bz_n \mid \bx_n), \nonumber
\end{align}
and find its minimizer
\begin{align}
    \argmin \widehat{\LLL}_{\text{SID}} = \argmin_{W = [\bw_1, \ldots, \bw_s] \in \mathbb{R}^{c \times s}} - \frac{1}{N} \sum_{n=1}^N \bz_n^{\top} \log \left(\text{softmax} (W \bx_n)\right). \label{eq:del_surrogate}
\end{align}
The expression \eqref{eq:del_surrogate} should be interpreted as a mechanism-isolating representation of the induced mapping from context states to output distributions: row $j$ of $W$ encodes the logits used when $\bx_n=\boldsymbol{e}_j$, i.e., $\bp(\cdot \mid \bx_n = \boldsymbol{e}_j)=\text{softmax}(W_{j,:})$. 

\begin{remark}{\textbf{(Why the surrogate problem is sufficient).}}
The matrix $W\in\mathbb R^{c\times s}$ should be interpreted as a reduced-form representation of the recommender’s conditional output behavior after one-hotting contexts. In the full encoder-decoder transformer, the cross-attention block (together with the remaining decoder layers) maps a history/prefix context to a continuous hidden state $\br \in \mathbb R^d$, and an output head then maps $\br$ to logits over the finite catalog items. Our analysis replaces this end-to-end mapping by a per-context categorical model $\bp(\cdot\mid \bx)=\mathrm{softmax}(W^\top \bx)$, where each one-hot context state $\bx=\boldsymbol{e}_j$ is allowed to have its own arbitrary logit vector $W_{j,:}$. This does not claim that the transformer literally implements an unconstrained table over contexts, but rather asks what happens in the best-case regime where the learner is expressive enough to represent an \emph{arbitrary} conditional distribution within each context state.
\end{remark}

\subsection{A Lemma on Empirical Conditional Frequencies}\label{pf:del_empirical}

For each context index $j \in [c]$, we define the index set and counts
\begin{align}
    C_j := \{n \in [N]: \bx_n = \boldsymbol{e}_j\}, \quad N_j:=|C_j|, \quad N_{jk}:= \sum_{n \in C_j} \mathbf 1\{\bz_n = \boldsymbol{e}_k\}. \nonumber
\end{align}
We assume $N_j>0$ for all $j \in [c]$, and let the empirical conditional distribution within context $j$ be
\begin{equation}\label{eq:del_empirical_p}
    \widehat{p}_{jk} := \frac{N_{jk}}{N_j} = \frac{1}{|C_j|} \sum_{n \in C_j} (\bz_n)_k, \quad \forall k\in[s].
\end{equation}
We show in this section that the cross entropy fitting \eqref{eq:del_surrogate} sets the conditional output distribution in each context state $\bp^\star(\cdot\mid \bx=\boldsymbol{e}_j)$ to the empirical conditional frequencies $\widehat{\bp}_j$.

\begin{lemma}[Empirical conditional frequencies]\label{lem:del_eq2}
The minimizer of \eqref{eq:del_surrogate} over conditional distributions satisfies
\begin{align}
    \bp^\star(\cdot\mid \bx=\boldsymbol{e}_j) = \widehat{\bp}_j \quad \forall j \in [c]. \nonumber
\end{align}
\end{lemma}

\begin{proof}{Proof of Lemma \ref{lem:del_eq2}}
Fix $n$ with $\bx_n=\boldsymbol{e}_j$. Then we have $W^\top \bx_n=W^\top \boldsymbol{e}_j = W^\top_{j,:} \in\mathbb R^s$, i.e., the $j$-th row of $W^\top$. Note that $\bp(\cdot \mid \bx_n)=\text{softmax}(W^\top_{j,:})$ is the same for all samples in $C_j$.
Writing $p_{jk}:=\text{softmax}(W^\top_{j,:})_k$, we have
\begin{align}
    \widehat{\LLL}_{\text{SID}} = -\frac{1}{N} \sum_{j=1}^c \sum_{n \in C_j} \sum_{k=1}^s (\bz_n)_k \log p_{jk} = \sum_{j=1}^c \frac{N_j}{N}\left(-\sum_{k=1}^s \widehat{p}_{jk}\log p_{jk}\right), \nonumber
\end{align}
where $\widehat{p}_{jk}$ is defined in \eqref{eq:del_empirical_p}. Hence minimizing $\widehat{\LLL}_{\text{SID}}$ reduces to minimizing, 
\begin{equation}\label{eq:del_min_L}
    \min_{\bp_j \in \Delta_s} \left(-\sum_{k=1}^s \widehat{p}_{jk} \log p_{jk}\right) \quad \forall j \in [c].
\end{equation}
We solve \eqref{eq:del_min_L} by the KKT condition. We write down the Lagrangian with $\lambda$ being the multiplier of $\sum^s_{k=1} p_{jk} =1$:
\begin{align}
    L(\bp, \lambda) = -\sum_{k=1}^s \widehat{p}_{jk} \log p_{jk} + \lambda\left(\sum^s_{k=1} p_{jk} - 1\right). \nonumber
\end{align}
Differentiate w.r.t. $\bp_j$:
\begin{align}
    \frac{\partial L(\bp, \lambda)}{\partial p_{jk}} = -\frac{\widehat{p}_{jk}}{p_{jk}} + \lambda = 0 \Rightarrow p_{jk} = \frac{\widehat{p}_{jk}}{\lambda}. \nonumber
\end{align}
Plugging the differentiation into $\sum^s_{k=1} p_{jk} =1$, we have
\begin{align}
    \sum^s_{k=1} \frac{\widehat{p}_{jk}}{\lambda} = 1 \Rightarrow \lambda  = \sum^s_{k=1} \widehat{p}_{jk} = 1. \nonumber
\end{align}
Therefore, the minimizer of $\widehat{\LLL}_{\text{SID}}$ is $$\bp^\star(\cdot\mid \bx=\boldsymbol{e}_j) = \widehat{\bp}_j.$$ If all entries $\widehat{p}_{jk}$ are positive, choose logits $(W^\top_{j,:})_k = \log \widehat{p}_{jk}$ (up to an additive constant), such that $\text{softmax}(W^\top_{j,:})=\widehat{\bp}_j$. If some $\widehat{p}_{jk}=0$, the minimizer $\widehat{\bp}_j$ lies on the boundary of the simplex. In this case, a sequence of logits with $W^\top_{j,k} \to -\infty$ for those $k$ achieves $\text{softmax}(W^\top_{j,:})\to \widehat{\bp}_j$. In both cases, our lemma holds.
$\hfill\square$
\end{proof}
Lemma \ref{lem:del_eq2} suggests that under cross-entropy training, the optimal predictor for a fixed context state $j \in [c]$ is the empirical histogram of next outcomes observed under that state $\widehat{\bp}_j$. This is intuitive and matches the maximum-likelihood objective. 

\subsection{Proof of Proposition \ref{prop:del_bias}}\label{pf:prop_del}

Let there be $R$ training rounds. Fixing a context $j$, we let $\bp^{(0)}_j = (p_{j1}, \ldots, p_{js})$ be the ground-truth distribution. We define 
\begin{align}
    \bp^{(r)}_j = (p^{(r)}_{j1}, \ldots, p^{(r)}_{js}) \qquad \forall 1 \leq r \leq R \nonumber
\end{align}
to be the $r$-th generation model trained on data from the last round. Lemma \ref{lem:del_eq2} shows that $\bp^{(r)}_j$ follows the empirical histogram generated by $\bp^{(r-1)}_j$. We train on a data set where, in round $r$, there are
\begin{itemize}
    \item \textit{Organic (real) data:} $N_j$ data points make decision based on their own judgment, characterized by i.i.d. one-hot tokens $\{\bz^{(0)}_{j,n}\}^{N_j}_{n=1}$ sampled following the ground-truth distribution $\bp^{(0)}_j$, i.e., $\Pr(\bz^{(0)}_{j,n} = \boldsymbol{e}_k)= p^{(0)}_{jk}$;
    \item \textit{Delegated (synthetic) data:} for $r \geq 2$, $\widehat{N}_j$ data points let agentic AI delegate their decisions, characterized by i.i.d. one-hot tokens $\{\bz^{(r-1)}_{j,n}\}^{\widehat{N}_j}_{n = 1}$ sampled following the $(r-1)$-th generation empirical distribution $\bp^{(r-1)}_j$, i.e., $\Pr(\bz^{(r-1)}_{j,n} = \boldsymbol{e}_k) = p^{(r-1)}_{jk}$.
\end{itemize}

Our proof is inspired by \citet{seddik2024bad}, where they show cross-entropy-induced model collapse. Since we train on the mixture, we have
\begin{align}
    \bp^{(r)}_j= \frac{1}{N_j + \widehat{N}_j}\left(\sum_{n=1}^{N_j} \bz^{(0)}_{j,n} + \sum_{n=1}^{\widehat{N}_j} \bz^{(r-1)}_{j,n}\right). \label{eq:del_pm}
\end{align}
Recall that
\begin{align}
    S^{(r)}_j:=\mathbb E\left[\|\bp^{(r)}_j\|_2^2\right], \quad S^{(0)}_j:=\|\bp^{(0)}_j\|_2^2. \nonumber
\end{align}
The larger $S^{(r)}_j$ is, the less diverse the generation output is. In the extreme case where $S^{(r)}_j \rightarrow 1$ as $r \rightarrow +\infty$, the output tends to be deterministic, and the model collapses. In the following proposition, we show how the empirical histogram (characterized by $S^{(r)}_j$) evolves.

Fix a $j \in [c]$. In the following, we neglect all subscripts $j$ in $N_j, \widehat{N}_j, \bp^{(r)}_j, \bz^{(r)}_{j,n}, S^{(r)}_j$ since the proof is the same for all context state $j$. We write $N':= N + \widehat{N}$, and define the second-moment quantity
\begin{align}
    \nu_k^{(r)} := \mathbb{E}\left[\left(p_k^{(r)}\right)^2\right] \qquad \forall k \in [s]. \nonumber
\end{align}
We first compute $\nu_k^{(1)}$. Since $p_k^{(1)} = (\sum^{N}_{n=1} (\bz^{(0)}_n)_k)/N$ and $(\bz^{(0)}_n)_k$ are i.i.d. Bernoulli($p_k^{(0)}$), we have
\begin{align}
    \nu_k^{(1)} =\mathbb{E}[(p_k^{(1)})^2] = \frac{1}{N^2} \cdot \mathbb{E}\left[\left(\sum_{n=1}^N (\bz^{(0)}_n)_k\right)^2\right] = \frac{p_k^{(0)}}{N}+\left(1-\frac{1}{N}\right)\left(p_k^{(0)}\right)^2. \nonumber
\end{align}
For $r \geq 2$, given \eqref{eq:del_pm}, we have
\begin{align}
    \nu_k^{(r)} = & \mathbb{E}\left[\left(\frac{1}{N'}\left(\sum_{n=1}^N (\bz^{(0)}_n)_k+\sum_{n=1}^{\widehat{N}} (\bz^{(r-1)}_n)_k\right)\right)^2\right] \nonumber \\
    = & \frac{1}{(N')^2} \mathbb{E}\left[\left(\sum_{n=1}^N (\bz^{(0)}_n)_k\right)^2 + \left(\sum_{n=1}^{\widehat{N}} (\bz^{(r-1)}_n)_k\right)^2 + 2\left(\sum_{n=1}^N (\bz^{(0)}_n)_k\right)\left(\sum_{n'=1}^{\widehat{N}} (\bz^{(r-1)}_{n'})_k\right)\right]. \label{eq:pf_del_nu}
\end{align}

\textbf{Step 1: the user-own-decision term.}
Using independence across $n$,
\begin{equation}\label{pf:del_real}
    \mathbb{E}\left[\left(\sum_{n=1}^N (\bz^{(0)}_n)_k\right)^2\right] = Np_k^{(0)}+(N^2-N)\left(p_k^{(0)}\right)^2.
\end{equation}

\textbf{Step 2: the agent-delegation term.}
Condition on $\bp^{(r-1)}$, we know $(\bz^{(r-1)}_1)_k, \ldots, (\bz^{(r-1)}_{\widehat{N}})_k$ are conditionally i.i.d. Bernoulli$(p_k^{(r-1)})$, so
\begin{align}
    \mathbb{E}\left[\left(\sum_{n=1}^{\widehat{N}} (\bz^{(r-1)}_n)_k \right)^2 \Bigg| \bp^{(r-1)}\right] = \widehat{N} p_k^{(r-1)}+(\widehat{N}^2-\widehat{N})(p_k^{(r-1)})^2. \nonumber
\end{align}
Since
\begin{align}
    \mathbb{E}\left[p_k^{(r-1)}\right] = & (1-\alpha) p_k^{(0)} + \alpha \mathbb{E}\left[p^{(r-2)}_k\right] \nonumber \\
    = &\left((1-\alpha) \cdot \frac{1-\alpha^{r-1}}{1-\alpha} + \alpha^{r-1}\right) p_k^{(0)} = p_k^{(0)}, \nonumber
\end{align}
we have
\begin{align}
    \mathbb{E}\left[\left(\sum_{n=1}^{\widehat{N}} (\bz^{(r-1)}_n)_k \right)^2\right] = & \mathbb{E}\left[\mathbb{E}\left[\left(\sum_{n=1}^{\widehat{N}} (\bz^{(r-1)}_n)_k \right)^2 \Bigg| \bp^{(r-1)}\right]\right] \nonumber \\
    = & \widehat{N} \mathbb{E}\left[p_k^{(r-1)}\right]+(\widehat{N}^2-\widehat{N}) \mathbb{E}\left[(p_k^{(r-1)})^2\right] \nonumber\\
    = & \widehat{N} p_k^{(0)} +(\widehat{N}^2-\widehat{N})\nu_k^{(r-1)}. \label{pf:del_sync}
\end{align}

\textbf{Step 3: the cross term.}
By symmetry and independence,
\begin{align}
    \mathbb{E}\left[\left(\sum_{n=1}^N (\bz^{(0)}_n)_k\right) \left(\sum_{n'=1}^{\widehat{N}} (\bz^{(r-1)}_{n'})_k\right)\right] = N \cdot \widehat{N} \cdot \mathbb{E}\left[(\bz^{(0)}_1)_k (\bz^{(r-1)}_1)_k\right]. \nonumber
\end{align}
Using the tower rule and $\mathbb{E}[(\bz^{(r-1)}_1)_k\mid \bp^{(r-1)}]=p_k^{(r-1)}$, we have
\begin{equation}\label{pf:del_cross_1}
    \mathbb{E}\left[(\bz^{(0)}_1)_k (\bz^{(r-1)}_1)_k\right] = \mathbb{E}\left[\mathbb{E}\left[(\bz^{(0)}_1)_k (\bz^{(r-1)}_1)_k \Big| \bz^{(0)}_1, \bp^{(r-1)}\right]\right] = \mathbb{E}\left[(\bz^{(0)}_1)_k p_k^{(r-1)}\right].
\end{equation}
For $r \geq 2$, by the definition of $p_k^{(r-1)}$, we have
\begin{align}
    \mathbb{E}\left[(\bz^{(0)}_1)_k p_k^{(r-1)}\right] = & \frac{1}{N'} \mathbb{E}\left[(\bz^{(0)}_1)_k \sum_{n=1}^{N} (\bz^{(0)}_{n})_k\right] + \frac{1}{N'} \mathbb{E}\left[(\bz^{(0)}_1)_k \sum_{n=1}^{\widehat{N}} (\bz^{(r-2)}_{n})_k\right] \nonumber \\
    = & \frac{1}{N'} \left(p_k^{(0)} + (N-1)\left(p_k^{(0)}\right)^2\right) + \frac{1}{N'} \mathbb{E} \left[(\bz^{(0)}_1)_k \sum_{n=1}^{\widehat{N}} (\bz^{(r-2)}_{n})_k\right]. \label{pf:del_cross_2}
\end{align}
Also we have
\begin{align}
    \mathbb{E} \left[(\bz^{(0)}_1)_k \sum_{n=1}^{\widehat{N}} (\bz^{(1)}_{n})_k\right] = & \widehat{N} \cdot \mathbb{E} \left[(\bz^{(0)}_1)_k p^{(1)}_k\right] \nonumber \\
    = &\frac{\widehat{N}}{N} \cdot \mathbb{E} \left[(\bz^{(0)}_1)_k  \sum^N_{n=1} (\bz^{(0)}_n)_k\right] \nonumber \\
    =& \frac{\widehat{N}}{N} \left(p_k^{(0)} + (N-1) \left(p_k^{(0)}\right)^2\right), \label{pf:del_cross_3}
\end{align}
and
\begin{align}
    \mathbb{E} \left[(\bz^{(0)}_1)_k \sum_{n=1}^{\widehat{N}} (\bz^{(r-2)}_{n})_k\right] = & \widehat{N} \cdot \mathbb{E} \left[(\bz^{(0)}_1)_k p^{(r-2)}_k\right] \nonumber \\
    = &\frac{\widehat{N}}{N + \widehat{N}} \cdot \mathbb{E} \left[(\bz^{(0)}_1)_k  \left(\sum^N_{n=1} (\bz^{(0)}_n)_k  + \sum^{\widehat{N}}_{n=1} (\bz^{(r-3)}_n)_k\right)\right] \nonumber \\
    =& \alpha \left[\left(p_k^{(0)} + (N-1) \left(p_k^{(0)}\right)^2\right) + \mathbb{E} \left[(\bz^{(0)}_1)_k \sum_{n=1}^{\widehat{N}} (\bz^{(r-3)}_{n})_k\right]\right], \label{pf:del_cross_4}
\end{align}
By induction, we have 
\begin{align}
    \mathbb{E}\left[(\bz^{(0)}_1)_k p_k^{(r-1)}\right] = &\cdot \frac{1- \alpha^{r-2}}{1 - \alpha} \frac{p_k^{(0)} + (N-1)\left(p_k^{(0)}\right)^2}{N'} + \alpha^{(r-2)} \cdot \frac{p_k^{(0)} + (N-1)\left(p_k^{(0)}\right)^2}{N} \nonumber \\
    = & \frac{p_k^{(0)} + (N-1)\left(p_k^{(0)}\right)^2}{N}. \label{pf:del_cross}
\end{align}

\textbf{Step 4: Solve $\nu^{(r)}_k$ and $S^{(r)}$.}
Plugging \eqref{pf:del_real}, \eqref{pf:del_sync} and \eqref{pf:del_cross} into \eqref{eq:pf_del_nu}, we have
\begin{align}
    \nu^{(r)}_k = \frac{(1-\alpha)(1+2\alpha)}{N} p_k^{(0)} + \left(1-\frac{1}{N}\right)(1-\alpha^2) \left(p_k^{(0)}\right)^2 + \alpha \left(\left(1+\frac{1}{N}\right) \alpha - \frac{1}{N}\right) \nu_k^{(r-1)}. \nonumber
\end{align}
Solving the $\nu^{(r)}_k$ iteration and plugging in $\nu^{(1)}_k = p_k^{(0)}/N + (1-1/N) (p_k^{(0)})^2$, we have for all $r \geq 1$:
\begin{equation}\label{pf:nu_closed}
    \nu_k^{(r+1)} = \frac{p_k^{(0)}}{N} \cdot \frac{1+2\alpha - (1 -\frac{1}{N}) \alpha \beta^{r}}{1 +(1+\frac{1}{N}) \alpha} + \left(1-\frac{1}{N}\right)\left(p_k^{(0)}\right)^2 \cdot \frac{ 1 + \alpha + \frac{\alpha}{N} \beta^{r}}{1+ (1+\frac{1}{N}) \alpha},
\end{equation}
where $\beta = \alpha((1+1/N)\alpha - 1/N)$. Solving \eqref{pf:nu_closed} over $k \in [s]$, with the fact that $\sum^s_{k=1} p_k^{(0)} = 1$ and $\sum^s_{k=1} (p_k^{(0)})^2 = S^{(0)}$, we have 
\begin{align}
    S^{(r+1)} = & \frac{1}{N} \cdot \frac{ 1 + 2\alpha - (1 - \frac{1}{N}) \alpha \beta^{r}}{1+ (1 + \frac{1}{N}) \alpha} + \left(1 - \frac{1}{N}\right) S^{(0)} \cdot \frac{ 1 + \alpha + \frac{1}{N} \alpha \beta^{r}}{1 + (1+\frac{1}{N}) \alpha} \nonumber \\
    = & \frac{ \frac{1}{N} (1 + 2\alpha) + (1 - \frac{1}{N})(1+\alpha) S^{(0)}}{1+ (1 + \frac{1}{N}) \alpha} - \frac{1}{N} \left(1 - \frac{1}{N}\right)\frac{\left(1 - S^{(0)}\right)\alpha \beta^{r}}{1 + (1+\frac{1}{N}) \alpha}.
\end{align}
Note that we assume $\widehat{N} > 1$, so that $\alpha > 1/(N+1)$ and $\beta \in (0,1)$. If $N>1$ and $S^{(0)}<1$, then $S^{(r)}$ is strictly increasing in $r$.
$\hfill\square$

\end{document}